\mathchardef\period=\mathcode`.
\documentclass[11pt]{article}
\usepackage{amsthm}
\usepackage{mathtools}
\usepackage{amssymb}

\usepackage[toc,page]{appendix}

\usepackage[utf8]{inputenc}

\usepackage{algorithm}
\usepackage{color}
\usepackage[english]{babel}
\usepackage{csquotes}

\usepackage{graphicx}
\usepackage{txfonts}
\usepackage[noend]{algpseudocode}
\usepackage{wrapfig,epsfig}
\usepackage{epstopdf}
\usepackage{url}
\usepackage{graphicx}
\usepackage{color}
\usepackage{epstopdf}
\usepackage{algpseudocode}
\usepackage{scrextend}
\usepackage{caption}
\usepackage[T1]{fontenc}
\usepackage{bbm}
\usepackage{comment}
\usepackage[inkscapeformat=png]{svg}

\usepackage{booktabs}
\usepackage{subcaption}
\usepackage{multirow}
\usepackage{numprint}
 \usepackage{float}

\usepackage{lineno}

\usepackage[
backend=biber,
style=alphabetic,
sorting=ynt
]{biblatex}

\usepackage{tikz}

\usepackage{cleveref}
\crefformat{section}{\S#2#1#3} 
\crefformat{subsection}{\S#2#1#3}
\crefformat{subsubsection}{\S#2#1#3}

\usetikzlibrary{arrows}
\usepackage[margin=1in]{geometry}
\graphicspath{{./figs/}}

\usepackage{tikz}
\usepackage{tikzit}

\tikzstyle{tree_node}=[fill=white, draw=black, shape=circle]
\tikzstyle{fill_node}=[fill=black, draw=black, shape=circle, minimum size=2pt, inner sep=0pt]
\tikzstyle{vertex_text_node}=[fill=white, draw=black, shape=circle, minimum size=2pt, inner sep=0pt]
\tikzstyle{leaf_node}=[fill=black, draw=black, shape=circle, minimum size=4pt, inner sep=0pt]

\tikzstyle{dash_lines}=[-, thick, dashed, dash pattern=on 0.6mm off 0.8mm, line width=0.6mm]
\tikzstyle{selected_edge}=[-, line width=0.45mm]
\tikzstyle{blue_edge}=[-, draw=blue]
\tikzstyle{dash_thin}=[-, dashed, dash pattern=on 0.2mm off 0.4mm, line width=0.25mm]
\tikzstyle{vertical_curly_braket}=[-, decorate, decoration={{brace,mirror,raise=5pt}}, line width=0.5mm]

\input{figs/styles.tikzdefs}
\usetikzlibrary{positioning}
\definecolor{processblue}{cmyk}{0.96,0,0,0}

\author{
Amit Suliman\thanks{Research supported by ERC Starting grant CODY 101039914, and ISF grant 3011005535. \texttt{tima.suliman@gmail.com}}\\
  The Hebrew University \\ 
  \and 
  Omri Weinstein\thanks{Supported by an ERC Starting grant CODY 101039914, and ISF grant 3011005535,  \texttt{omriwe@cs.huji.ac.il}}\\
  The Hebrew University and Columbia}

\newcommand{\cK}{\mathcal{K}}

\newcommand{\x}{\boldsymbol{x}}

\newcommand{\linf}{\ell_{\infty}}

\newtheorem{theorem}{Theorem}[section]

\newtheorem{lemma}[theorem]{Lemma}

\newtheorem{conjecture}[theorem]{Conjecture}

\newtheorem{claim}[theorem]{Claim}

\theoremstyle{definition}
\newtheorem{definition}[theorem]{Definition}
\newtheorem{fact}[theorem]{Fact}

\theoremstyle{definition}
\newtheorem{remark}[theorem]{Remark}

\newcommand{\amit}[1]{{\color{red} Amit: #1}}

\newcommand{\bI}{\mathbf{I}}
\newcommand{\buni}{\boldsymbol{1}}

\newcommand{\wh}{\widehat}
\newcommand{\wt}{\widetilde}

\newcommand{\eps}{\epsilon}

\newcommand{\R}{\mathbb{R}}

\renewcommand{\varepsilon}{\epsilon}
\renewcommand{\tilde}{\wt}
\renewcommand{\hat}{\wh}

\DeclareMathOperator{\Tr}{\textbf{Tr}}

\DeclareMathOperator{\rank}{rank}

\DeclareMathOperator{\diag}{diag}

\DeclareMathOperator{\diam}{diam}

\DeclareMathOperator{\logdet}{logdet}
\DeclareMathOperator{\LSE}{LSE}
\DeclareMathOperator{\SEV}{SEV}

\newcommand{\twopartdef}[4]
{
	\left\{
		\begin{array}{ll}
			#1 & \mbox{if } #2 \\
			#3 & \mbox{if } #4
		\end{array}
	\right.
}

\usepackage{multicol}

\usepackage{footnote}
\newcommand{\numberthis}{\refstepcounter{equation}\tag{\theequation}}
\numberwithin{equation}{section}

\DeclareCaptionFormat{myformat}{#1#2#3\hrulefill}
\begin{document}

\date{}

\title{
Infinite Lewis Weights in  Spectral Graph Theory}



\begin{titlepage}
     \maketitle

\begin{abstract}

We study the spectral implications of re-weighting a graph by the $\linf$-Lewis weights of its edges.  
 Our main motivation is the ER-Minimization problem (Saberi et al., SIAM'08): Given an undirected graph $G$, the goal is to 
 find positive normalized edge-weights $w\in \R_+^m$ which minimize the sum of pairwise \emph{effective-resistances} of $G_w$ (Kirchhoff's index). By contrast, $\linf$-Lewis weights minimize the \emph{maximum} effective-resistance of \emph{edges}, but are much cheaper to approximate, especially for Laplacians. 
With this algorithmic 
motivation, we study the ER-approximation ratio obtained 
by Lewis weights. 

Our first main result is 
that $\linf$-Lewis weights provide a constant ($\approx 3.12$) approximation for ER-minimization on \emph{trees}. 
The proof introduces a new technique, a local polarization process for   effective-resistances ($\ell_2$-congestion) on trees, which is of independent interest in electrical network analysis. 
For general graphs, we prove an upper bound $\alpha(G)$ on the approximation ratio obtained by Lewis weights, which is always $\leq \min\{ \text{diam}(G), 
\kappa(L_{w_\infty})\}$, where $\kappa$ is the condition number of the weighted Laplacian. 
All our approximation algorithms run in \emph{input-sparsity} time $\tilde{O}(m)$, a major improvement over Saberi et al.'s $O(m^{3.5})$ SDP for exact ER-minimization. 
We conjecture that Lewis weights provide an  $O(\log n)$-approximation for \emph{any} graph, and show experimentally that the ratio on many graphs is $O(1)$. 

Finally, we demonstrate the favorable effects of $\linf$-LW reweighting on the \emph{spectral-gap} of graphs  
and on their \emph{spectral-thinness} (Anari and Gharan, 2015).
En-route to our results, 
we prove  a weighted analogue of Mohar's classical bound on $\lambda_2(G)$, and provide a new characterization of leverage-scores of a matrix, as the gradient (w.r.t weights) of the volume of the enclosing ellipsoid.

     \end{abstract}
     \thispagestyle{empty}
 \end{titlepage}

\tableofcontents

\newpage 
\section{Introduction}

The $\ell_p$-Lewis weights of a matrix $A \in \R^{m\times n}$ generalize the notion of statistical \emph{leverage scores} of rows, 
which informally quantify the importance of row $a^\top_i$ in composing the spectrum of $A$. Formally, the $\ell_p$-Lewis weights ($\ell_p$-LW) of $A$ are the unique vector $\overline{w} \in \R^m$ satisfying the following \emph{fixed-point} equation
\begin{align}\label{eq_Lp_LW_def}
a_i^T\left(A^T \overline{W}^{1-2 / p} A \right)^{-1} a_i=\overline{w}_i^{2 / p}  \;\;\;\;, \;\;\;\;\; \forall \; i \in [m]
\end{align}
where $\overline{W} = \diag(\overline{w}) \in \R^{m\times m}$. 
This definition of \cite{Lewis78} can be viewed as a \emph{change-of-density} of rows, that requires each of the reweighted rows to have leverage score $1$, 
i.e., that the $i$th row of the reweighted matrix 
$\overline{W}^{1/2 - 1/p} A$ should end up with leverage score $\overline{w}_i$,
hence the ``fixed point" condition  \eqref{eq_Lp_LW_def}. 
$p$-norm Lewis weights have found important algorithmic applications in optimizatin and random matrix theory in recent years, from dimensionality-reduction for $p$-norm regression (Row-Sampling \cite{CP14, flps21}) to the design of optimal barrier functions for linear programming \cite{LS14}, and spectral sparsification for Laplacian matrices  (\cite{SS11, KMP16}). 
For $p=\infty$, which is the focus of this paper, the $\linf$-LW have a particularly nice geometric interpretation -- the ellipsoid  
$\mathcal{E} = \{ x \in \R^n \mid x^T (A^T \overline{W} A)^{-1} x \leq 1 \} $ 
defined by $\linf$-LW$(A)$ is the 
\emph{dual} solution of the 
minimum-volume ellipsoid enclosing the pointset $\{ a_i \}_{i=1}^m $, known as the the 
Outer 
John Ellipsoid: 
\begin{align}\label{equ_John_Ellipsoid}
\max\limits_{M \succeq 0} \left\{ \log\det M \; : a_i^T M a_i \leq 1, \; \forall i \in[m] \right\} .
\end{align}

In this paper we investigate $\linf$-LW as a tool for design and analysis of \emph{Laplacian} matrices. 
Since leverage-scores ($\ell_2$-LW) play a key role in spectral graph algorithms \cite{st04, SS11,KMP16}, it is natural to ask what role do higher-order Lewis weights play in graph theory and network analysis: 
\begin{quote}    
\emph{What are the 
spectral implications of reweighting a graph by the $\linf$-LW of its edges? }
\end{quote}
In order to explain this motivation and the role of $\linf$-LW in spectral graph theory, it is useful to interpret  the optimization Problem~\eqref{equ_John_Ellipsoid} as a special case of \emph{Experiment Optimal Design} \cite{boyd_convex_2004},
where the goal is to find 
an optimal convex combination of fixed linear measurements $\{a_i\}_{i=1}^m$, inducing a maximal-confidence ellipsoid for its least-square estimator, where ``maximal" is with respect to \emph{some} partial order on Positive-semidefinite matrices (e.g Loewner order, see \cref{appendix_optimal_design}).
A natural choice for such order is the \emph{volume}  (determinant) of the confidence ellipsoid, known as \emph{D-optimal design}:
Given the experiment matrix, $\{ a_i \in \R^n \}_{i=1}^m$, the primal problem of \eqref{equ_John_Ellipsoid} is  
\begin{align}\label{equ_D_optimal_def}
\begin{array}{ll}
\text { minimize } & \logdet \tag{D-optimal Design} \left(\sum_{j=1}^m g_j a_j a_j^T\right)^{-1}  \\
\text { subject to } & \buni^T g = 1 \ , \ g \geq 0.
\end{array}
\end{align}
Another natural order on PSD matrices is the matrix \emph{trace}, commonly known as \emph{A-optimal design} \cite{boyd_convex_2004}:
\begin{align}\label{equ_A_optimal_def}
\begin{array}{ll}
\text { minimize } & \Tr \left(\sum_{j=1}^m g_j a_j a_j^T\right)^{-1}  \tag{A-optimal Design} \\
\text { subject to } & \buni^T g = 1 \ , \ g \geq 0, 
\end{array}
\end{align}
which is a Semidefinite program (SDP). In both of these convex problems, the optimization is over positive normalized \emph{weight vectors} $g \in \R^m$, defining enclosing ellipsoids of the pointset $\{a_i\}_{i=1}^m$. 
Understanding the relationship between \eqref{equ_D_optimal_def} and   \eqref{equ_A_optimal_def}  is a central theme of this paper. 
The two different objective functions 
have a natural geometric interpretation:  
\eqref{equ_D_optimal_def} seeks to minimize the 
\emph{geometric mean} of the ellipsoid's semiaxis lengths, whereas \eqref{equ_A_optimal_def} seeks to minimize their \emph{harmonic mean}. 
Intuitively, this HM-GM viewpoint renders the problems quite different, and indeed we show that for \emph{general} PSD matrices,  there's an $\Omega(n/\log^2 n)$ separation between the two objectives:
\begin{theorem}[Informal]\label{thm_natural_gap}
    There are $n \times n$ PSD experiment matrices $A \succeq 0$ for which the D-optimal solution ($\linf$-LW($A$)) is no better than  an $\Omega(n/\log^2 n)$-approximation to the A-optimal design problem w.r.t $A$. 
\end{theorem}

This lower bound implies that a nontrivial relation between the two minimization problems can only hold due to some special structure of PSD matrices. As mentioned earlier, the focus of this paper is on \emph{graph Laplacians}. Our main motivation for studying A-vs-D optimal design for Laplacian matrices comes from network design applications, as in this case both problems correspond to 
controlling the \emph{distribution of electrical flows}   
and \emph{Effective Resistances} of the underlying graph. We now explain this correspondence. 


\paragraph{Electric Network Design:
Minimizing Effective Resistances on Graphs}
Given an undirected graph $G$, the \emph{effective resistance} $R_{ij}(G)$ is the electrical potential difference that appears across terminals $i$ and $j$ when a unit current source is applied between them (see Definition \eqref{def_ER}).
Intuitively, the ER between $i$ and $j$ is small when there are many paths between the two nodes with high conductance edges, and it is large when there are few paths, with lower conductance, between them.
Effective resistances have many  applications in electrical network analysis, spectral sparsification \cite{SS11} and Laplacian linear-system solvers \cite{KMP16}, maximum-flow algorithms \cite{CKM+11}, the commute and cover times in Markov chains \cite{CRR+97, Mat88}, continuous-time averaging networks,  finding thin trees \cite{AO15}, and in generating random spanning trees \cite{KM09, MST15, DKP+17}. 
Apart of all these, the distribution of effective resistances of a graph are natural graph properties to be investigated on their own right, as advocated by \cite{anari17, saberi}.

In a seminal work, Ghosh, Boyd and Saberi \cite{saberi} introduced the  \emph{Effective Resistance Minimization Problem} (ERMP): Given an unweighted graph $G$, find nonnegative, normalized edge-weights that minimize the total sum of effective resistance over all pairs of vertices of the reweighted graph $G_g$, also known as the \emph{Kirchhoff index} $\cK_g(G)$ \cite{Lukovits1999}. This problem can be expressed as the following Semidefinite Program:

\begin{equation}\label{equ_ERMP_def}
\begin{array}{ll@{}ll}
\text{minimize}  & \cK_g(G) := \sum\limits_{i,j} R_{ij} &\\
\text{subject to}&  g \geq 0 , \tag{ERMP} \ \boldsymbol{1}^T g=1
\end{array} 
\end{equation}
It is straightforward to verify that the ERMP is an \emph{A-optimal design problem over Laplacians} (see equation~\eqref{equ_R_tot_A_design}). This formulation of ERMP has the following interpretation -- we have real numbers $x_1,...,x_n$ at the nodes of $G$, which have zero sum. Each edge in $G$ corresponds to a possible `measurement', which yields the difference of its adjacent node values, plus noise. Given a fixed budget of measurements for estimating  $\boldsymbol{x}$, the problem is to select the \emph{fraction} of  experiments that should be devoted to each edge measurement. The optimal fractions are precisely the ERMP weights. 
The motivation of \cite{saberi} for studying this problem was improving the ``electrical connectivity" and mixing time of the underlying network. Minimizing resistance distances and electrical flows is an important primitive in routing networks (and molecular chemistry \cite{Lukovits1999}, where the problem in fact originated), and can be viewed as a \emph{convex} proxy for maximizing the spectral-gap and communication throughput of the reweighted network. As such, ERMP can be viewed as a network configuration algorithm \cite{HAAKSM13,MAK+20}. 

In \cite{saberi}, the authors design an ad-hoc interior-point method (IPM) which solves  ~\eqref{equ_ERMP_def} in $O(m^{3.5})$ time (more precisely, in $\sim \sqrt{m}$ generic linear-equations involving \emph{pseudo-inverses} of Laplacians, which are no longer SDD). We note that even with  recent machinery of \emph{inverse-maintenance} for speeding-up general IPMs \cite{jkl+20, HJ0T022}, one could at best 
 achieve $\tilde{O}(mn^{2.5} + n^{0.5}m^\omega)$, 
 but such acceleration (for sparse graphs) is only  of theoretical value. In either case, the above runtimes are quite daunting for large-scale networks, let alone for dense ones ($m\gg n$).  
This motivates resorting to  \emph{approximate} algorithms. Unfortunately, using black-box approximate SDP solvers (projection-free algorithms a-la \cite{AK16, Haz08}) seems inherently too slow for the ERMP SDP, due to its large \emph{width}  and the additive-guarantees of MwU.\footnote{The \emph{width} of the ERMP SDP is $\geq \Omega(n^2)$, hence  
approximate SDP solvers based on multiplicative-weight updates \cite{Haz08, AK16} 
require at least $mn^2$ time to 
achieve nontrivial approximation. 
Moreover, Hazan's algorithm only guarantees an \emph{additive} approximation $\eps\|C\|_F$ with respect to the Frobenius norm of the objective matrix, which in our case is $\Omega(n^{1.5})$, see~\cref{appendix_ermp_sdp}.}

By contrast, $\linf$-Lewis weights are much cheaper to 
approximate, especially for Laplacians -- 
While \emph{high-accuracy} 
computation of 
the John ellipsoid  \eqref{equ_John_Ellipsoid} is no faster than 
the aforementioned  IPM ($\sim m^{3.5}\log(m/\eps)$ \cite{Nem99}), 
low-precision approximation of $\linf$-LW turns out to be dramatically cheaper -- \cite{ccly,CP14} gave a simple and practical approximation algorithm to arbitrarily small precision $\eps$, via $\tilde{O}(1/\eps)$ repeated \emph{leverage-score} computations, which in the case of Laplacians can be done in input-sparsity time \cite{st04,KMP16}. 
Fortunately, low-accuracy approximation of $\linf$-LW suffices for ERMP  (see \cref{sec_computing_LW}). 
We remark that very recently, \cite{flps21} gave a \emph{high-accuracy} algorithm for computing $\ell_p$-LW using only $\tilde{O}(p^3 \log(1/\eps))$ 
leverage-score computations (Laplacian linear systems in our case), hence setting $p=n^{o(1)}$ yields an alternative $O(m^{1+o(1)}\log(1/\eps))$ time approximate algorithm, which is good enough for most applications \cite{flps21}.

Compared to the ERMP objective, the $\linf$-LW of the graph Laplacian $L(G)$ 
can be shown to minimize the \emph{maximal} effective-resistance over  \emph{edges} of $G$.  In fact, we show something stronger: The ER of edges in $G$ are the gradient, w.r.t weights $g\in \R^m$, of the \emph{volume} 
of the ellipsoid induced by the weighted Laplacian $L_g$. To the best of our knowledge this provides a new geometric characterization of the ER of a graph (and more generally, of statistical leverage-scores of a general matrix, see~\cref{sec_new_char_LS_ER}) : 

\begin{lemma} \label{lem_ER_char}
     $ER(G_g) = - \nabla_g \log\det L_g^+ .$ 
\end{lemma}
 The above discussion and the efficiency of $\linf$-LW approximation, raise the following natural question: 
\emph{How well do the $\linf$-Lewis weights of a graph approximate the Kirchhoff index (optimal ERMP weights) ?}  

The main message of this paper  is that $\linf$-LW rewighting of edges has various favorable effects on the spectrum of graphs, and provides a an efficient ``preprocessing" operation on large-scale (undirected) networks, well beyond the Kirchhoff index.   
We summarize our main findings  
in the next subsection.



\subsection{Main Results}

Let $\alpha_{A,D}(G)$ denote the approximation ratio obtained by $\linf$-LW for the ERMP objective 
(Definition \eqref{def_alpha_AD}). 
Our first main result is that $\alpha_{A,D}$ is constant for 
\emph{trees}: 
\begin{theorem}\label{thm_LW_apx_ERMP_trees}
    $\linf$-LW are a $3.12$-approximate solution for the ERMP problem on trees.
\end{theorem}
It is noteworthy that there are trees for which $\alpha_{A,D} > 2.5$ (see \cref{sec_experiments}), so the bound is nearly tight.    The proof of Theorem~\ref{thm_LW_apx_ERMP_trees} relies on a new technique, designated for trees, which may be of independent interest in electric network analysis -- we show that we can always locally modify the tree in a way that increases the approximation ratio $\alpha_{A,D}(T)$. More precisely, our proof introduces a (finite) ``polarization process" of ERs based on $\ell_2$-congestion of trees, that can be repeatedly  applied to yield the worst case ``polarized" family of trees (the \emph{Bowtie graph}, see Figure \ref{fig_tps_tree}). We provide a high-level overview of the proof in Section \ref{sec_tree_TOV}. 

Our second main result is an upper bound on the approximation ratio of $\linf$-LW for general graphs.  
The precise upper bound we develop is a function of the spectral parameter  
$ \alpha_1(g) := \frac{2}{(n-1)^2} \Tr L_g^{+}$ of the weighted laplacian of $G$, which is a central quantity in our analysis. Denoting by $g_{lw}$ the $\linf$-LW edge weights, define 
\begin{align}    \label{eq_alpha}
\alpha_{min}(G) := 
\min \left\{ \alpha_1(g_{\ell w})\; , \;  \left\lVert -\nabla_g \left(\log \alpha_1(g_{\ell w})\right) \right\rVert_\infty \right\}. 
\end{align}
For intuition, note that $\min\{x,(\log x)'\} = 1$,  so the above quantities tend to be anti-monotone in each other. 
\begin{theorem}[Upper Bound for General Graphs  ]\label{thm_UB_general_graphs}
For any undirected graph $G$,
\[
\alpha_{A,D}(G) \leq \alpha_{min}(G).  \]
\end{theorem}
We prove that $\alpha_{min}(G)$ is always at most   
$\leq    \min \left\{ \text{diam}(G), 
\kappa(L_{g_{\ell w}})\right\}$, where  
$\text{diam}(G)$ is the diameter of $G$, and  
$\kappa(L_{g_{\ell w}})$ is the \emph{condition-number} of its $\linf$-LW-weighted Laplacian. 
While this already gives a good approximation for low-diameter graphs, we stress that $\alpha_{min}$ typically provides a much tighter upper bound on $\alpha_{A,D}$: The 
\emph{lollipop} graph (clique of size $n$ connected to a path of length $n$) has both diameter and LW-condition-number $\Omega(n)$, but simulations show that $\alpha_{min}(\text{lollipop}_n) \leq O(\log n)$. In \cref{sec_experiments}, we showcase the approximation ratio $\alpha_{min}$ for various different graph families,  
and show that in practice it grows as $\tilde{O}(1)$.  
In light of Theorems \ref{thm_LW_apx_ERMP_trees}, \ref{thm_UB_general_graphs}, and our empirical evidence, we conjecture that $\linf$-LW provide an $O(\log n)$ approximation for the Kirchoff index of any graph: 

\begin{conjecture}[Lewis meets Kirchoff]\label{conj_LW_apx_ERMP} 
$\forall G$, \;  
    $\alpha_{A,D}(G) \leq O(\log n)$. 
\end{conjecture}

A stronger conjecture would 
be $\alpha_{\min}(G) \leq O(\log n)$; In fact, both our analysis and experiments indicate that this stronger conjecture holds (See \cref{sec_experiments}).   
We do not have a super-constant separations between $\alpha_{\min}(G)$ and $\alpha_{A,D}(G)$, and whether $\alpha_{\min}(G)$ \eqref{eq_alpha} is a tight upper bound is an intriguing question.

\paragraph{Spectral Implications of $\linf$-LW}
 Our third set of results demonstrates the favorable effects of $\linf$-LW reweighting on the eigenvalue distribution of graph Laplacians. We present two fundamental results in this direction:   

\begin{theorem}[$\linf$-LW reweighting and Mixing-time, Informal]
\label{thm_spectral_LW}
\;
    \begin{enumerate} 
        \item We generalize the classic spectral-gap bound of \cite{Mohar91} $(\lambda_2 \geq 4/nD)$ to \emph{weighted} graphs: 
        $\lambda_2(G_w) \geq  2/ (nD\cdot R_{\max}(G_w))$, where $R_{\max}$ is the maximal effective resistance of \emph{an edge} in the weighted graph. 
        \item We show a close connection between the optimality condition of the smooth-max-eigenvalue of a graph \cite{Nesterov03} ($\max_g \log (\Tr e^{L_g^+})$) and the fixed-point condition for $\linf$-LW \eqref{eq_Lp_LW_def}, implying a condition under which $\linf$-LW reweighting increases the softmax function.   
    \end{enumerate}
\end{theorem}

\paragraph{Spectrally-Thin Trees \cite{AO15}} 
In an influential result, \cite{AGM10} presented a rounding scheme for the Asymmetric TSP problem that breaks the $O(\log n)$ approximation barrier, assuming the underlying graph admits a  \emph{spectrally-thin} tree (STT) ($L_T \preceq \gamma L_G$, where $\gamma = \gamma_G(T)$ is the spectral thinness of the tree) . Anari et al. \cite{AO15} showed that for any graph, $\gamma_G(T)$ is at least $R_{\max_e}(G)$ and that an optimal tree with thinness $\tilde{\Theta}(R_{\max_e}(G))$ can be found in polynomial time. We generalize this result in the following way -- For any graph $G$, we can find a \emph{weighted} tree $T_w$ with spectral thinness of $L_{T_w} \preceq O((n-1)/m)L_G$:
\begin{lemma}
    For any connected graph $G = \langle V,E \rangle$ there is a weighted spanning tree $T_w$ such that $T_w$ has spectral thinness of $((n-1)/m) \cdot  O(\log n / \log \log n)$.
\end{lemma}
We show this is always at least as good as the \cite{AO15} tree. It would be interesting to see if this weighted version of STTs can be used for the ATSP rounding primitive.   
We believe that the host of our results will have further applications in the design and analysis of spectral graph algorithms. 





\section{Thechnical Overview}

Here we provide a high-level overview of the main new ideas required to prove Theorems \ref{thm_LW_apx_ERMP_trees} and  \ref{thm_UB_general_graphs}. 

\paragraph{ER Polarization: Technical Overview of Theorem \ref{thm_LW_apx_ERMP_trees} }\label{sec_tree_TOV}

Our proof technique for bounding the ERMP approximation ratio on trees ($\alpha_{A,D}(T)$) is based on the following important observation: the optimal weight of the $l$'th edge of a tree, $g^*_l(T)$, is proportional to its congestion, $c_l(T)$~\eqref{equ_opt_g_trees}. Therefore the optimal weight will be "farther away" from uniform (hence far from the LW solution \eqref{clm_LW_tree}) when the tree has a severe "bottleneck" (i.e., an edge crossed by many paths). With this observation, we show that given any tree $T$, we can
construct an alternative tree $T'$ whose approximation ratio is worse, i.e., $\alpha_{A,D}(T') > \alpha_{A,D}(T)$. Our construction is iterative, leveraging the above observation -- at each iteration we will increase the ``bottleneck" of the tree, until we reach a `fixed point' (the Bowtie graph). To formally define this polarization-process, we introduce the notion of \emph{Local Transformations}\footnote{The idea of LT is inspired by Elementary operations on matrices, and indeed they play a similar role in some sense.} (LT) on a tree. An LT is a local graph operation that changes the congestion of a \emph{single} edge of the tree. Formally, we denote by $E_k$ an LT such that
\begin{equation}\label{equ_LT_def}
    c_l(E_k \circ T) = c_l(T) , \ \forall \ l \neq k ,
\end{equation}
where $T' = E_k \circ T$ is the tree after the transformation. 

There are two natural ways to increase the bottleneck of a tree (i.e., make the distribution in~\eqref{equ_opt_g_trees} further from uniform): 
(1) Reduce the number of different paths, or  (2) Add more entry points (i.e leaves). It turns out that  these two operation can be formulated as an LT. This allows us to gradually construct the new tree with repeated use of LT. For this purpose, we shall define an \emph{upper LT} and a \emph{lower LT}, denoted $E^{\uparrow}_k$ and $E^{\downarrow}_k$ respectively, such that for any tree $T$, we have that:
\begin{equation}\label{equ_upper_lower_LT}
    \begin{split}
        &c_k(E^{\uparrow}_k \circ T) > c_k(T), \\  
        &c_k(E_k^{\downarrow} \circ T) < c_k(T).
    \end{split}
\end{equation}
(This definition will become clear below). 
The upper LT will correspond to (1) above, i.e., reducing the number of paths, and lower LT will correspond to (2) above, i.e., adding another entry point, which is consistent with our initial intuition. 
It turns out that the exact threshold for choosing which local operation to perform next is the norm of the optimal weight -- $||g^*(T)||_2^2$ :  
\begin{claim}\label{clm_apx_ratio_ET}
    Let $T$ be tree of order $n$, $E_k$ be an LT. Then $\alpha_{A,D}(T) \leq \alpha_{A,D}(E_k \circ T)$ if one of the following holds:
    \[
    E_k \text{ is a lower LT } \ \& \  g_k^*(T) \leq ||g^*(T)||_2^2 
    \]
    \qquad \qquad \qquad \qquad \qquad \quad or,
    \[
    E_k \text{ is an upper LT } \ \& \ g_k^*(T) > ||g^*(T)||_2^2 
    \]
\end{claim}

Following this claim, we divide the edges of the tree to two sets:
\[
E_<(T) \coloneqq \{ l \ \mid \ g_l^*(T) \leq ||g^*(T)||_2^2 \ \} \ , \ E_>(T) \coloneqq \{ l \ \mid \ g_l^*(T) > ||g^*(T)||_2^2 \ \}
\]
and operate with the appropriate LT on each set (until saturation -- $E_k \circ T = T$). This process is possible due to a key feature of these sets - they are \emph{invariant under LTs}.  More precisely, we prove that  
    for any tree $T$, and $k \in E_>$ we have that, \( E_>(E_k^{\uparrow} \circ T) = E_>(T)\), and vice versa for $E_<$.

We apply this process repeatedly (each iteration choosing the `right' LT), and  
show that that it must terminate after a finite number of steps, in a ``fixed-point" tree 
whose congestion is as ``polarized" as possible -- we call this family of trees \emph{Bowtie graphs} of some order (see \Cref{fig_tps_tree}).  We prove that Bowtie graphs maximize the ratio $\alpha_{A,D}(T)$ over the set of trees, and then bound the latter quantity for Bowtie graphs (by $\approx 3.12$) using the tools we develop for general graphs, on which we elaborate in the next paragraph.  

\begin{figure}[H]
    \centering
    \begin{tikzpicture}
	\begin{pgfonlayer}{nodelayer}
		\node [style={vertex_text_node}] (0) at (-5.25, 1.25) {\scriptsize $u_1$};
		\node [style={vertex_text_node}] (1) at (-2.75, 1.25) {\scriptsize $u_2$};
		\node [style={vertex_text_node}] (3) at (-0.25, 1.25) {\scriptsize  $u_3$};
		\node [style={vertex_text_node}] (5) at (3.75, 1.25) {\scriptsize  $u_p$};
		\node [style={fill_node}] (14) at (-7.25, 2.75) {};
		\node [style={fill_node}] (15) at (-7.25, 1.5) {};
		\node [style={fill_node}] (16) at (-7.25, -0.25) {};
		\node [style=none] (17) at (-7.75, 3) {\tiny $x_1$};
		\node [style=none] (18) at (-7.75, 1.5) {\tiny $x_2$};
		\node [style=none] (19) at (-7.75, -0.5) {\tiny $x_t$};
		\node [style={fill_node}] (20) at (5.775, 2.75) {};
		\node [style={fill_node}] (21) at (5.75, 1.5) {};
		\node [style={fill_node}] (22) at (5.75, -0.25) {};
		\node [style=none] (23) at (6.25, 3) {\tiny $y_1$};
		\node [style=none] (24) at (6.25, 1.75) {\tiny $y_2$};
		\node [style=none] (25) at (6.25, 0) {\tiny $y_s$};
	\end{pgfonlayer}
	\begin{pgfonlayer}{edgelayer}
		\draw [style={selected_edge}] (0) to (1);
		\draw [style={selected_edge}] (3) to (1);
		\draw [style={blue_edge}] (14) to (0);
		\draw [style={blue_edge}] (15) to (0);
		\draw [style={blue_edge}] (16) to (0);
		\draw [style={dash_thin}] (15) to (16);
		\draw [style={dash_thin}] (3) to (5);
		\draw [style={blue_edge}] (5) to (20);
		\draw [style={blue_edge}] (21) to (5);
		\draw [style={blue_edge}] (5) to (22);
		\draw [style={dash_thin}] (21) to (22);
	\end{pgfonlayer}
\end{tikzpicture}
    \caption{A Bowtie graph - $\mathcal{B}_{t,p,s}$. A path of length $p$ joined with stars of size $t,s$ on both sides}
    \label{fig_tps_tree}
\end{figure}
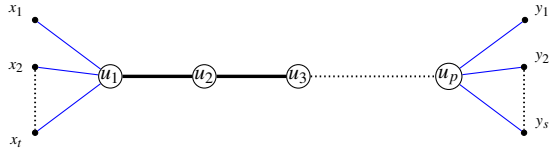

\paragraph{Overview of the General Upper Bound (Theorem \ref{thm_UB_general_graphs})}

There is a simple intuitive explanation for  
why $\linf$-LW($G$) provide an $O(\text{diameter})$-approximation 
to the ERMP problem on any graph $G$, i.e.,  $\alpha_{A,D} \leq \diam(G)$. Indeed, since 
LW minimize the maximal ER among \emph{edges} of $G$ (see next section), and 
since effective-resistances are well-known to form a \emph{metric}, the triangle-inequality implies that $ER(i,j)$ for any pair of vertices in $G_{lw}$ is at most the sum of ERs along the edges of a \emph{shortest path} between $i$ and $j$, which  is at most $\diam(G_{lw}) = \diam(G)$.  We use an AM-GM argument to prove a stronger bound: For any $G$, 
\begin{equation}\label{eq_alpha_1}
    \alpha_{A,D}(G) \leq  \alpha_1(g_{\ell w}),
\end{equation} 
where $ \alpha_1(g) := \frac{2}{(n-1)^2} \Tr L_g^{+}$ is a parameter closely related to the optimal solution $\cK_G(g)$ -- Indeed, the harmonic-mean characterization mentioned earlier in the introduction implies that the optimal ERMP value can be rewritten as (see Equation~\eqref{equ_R_tot_rep_tr})  
\[ 
    \cK(g^*) = n \Tr L_{g^*}^+ =
    \frac{n(n-1)}{HM(\lambda(L_{g^*}))},  
\]
where $HM(\lambda(L_{g^*}))$ is the harmonic mean of the $n-1$ (positive) eigenvalues of the weighted Laplacian with the optimal ERMP weights.  
With some further manipulations, we can use the AM-GM inequality to show that the approximation ratio $\cK(g_{lw})/\cK(g^*) \leq \alpha_1(G)$, which is always at most $\diam(G)$ via the triangle inequality, but is typically smaller, see the experiments section. 

What if the diameter of $G$ is large? In this case we show  that a ``\emph{dual}" function of $\alpha_1(G)$ is typically small. Specifically, we prove that the following quantity is also an upper bound on the  approximation ratio: 
\begin{equation}\label{eq_alpha_2}
    \alpha_{A,D}(G) \leq \left\lVert -\nabla_g \left(\log \alpha_1(g_{\ell w})\right) \right\rVert_\infty. 
\end{equation} 
This bound is somewhat more technical and less intuitive, but the key for deriving it is based on the ERMP \emph{duality gap} of \cite{saberi}, which naturally involves gradient-suboptimality conditions with respect to weights. Using calculus and the Courant-Fisher principle we then show that this is at most $\kappa(L(G_{lw}))$. 

Combining \eqref{eq_alpha_1} and \eqref{eq_alpha_2}  yields Theorem \ref{thm_UB_general_graphs} :
$\alpha_{A,D}(G) \leq  \min \left\{ \alpha_1(g_{\ell w})\; , \;  \left\lVert -\nabla_g \left(\log \alpha_1(g_{\ell w})\right) \right\rVert_\infty \right\}$.  
As mentioned earlier in the introduction, the intuition for why the \emph{minimum} of the two aformentioned quantities should always be small comes from the scalar inequality  $\min\{x,(\log x)'\} = 1$. Indeed, all our simulations corroborate that the minimum is $\tilde{O}(1)$, hence $\alpha_{min}(G)$ is typically a much tighter bound than $\min\{\diam(G), \kappa(L(G_{lw}))\}$. Conjecture \ref{conj_a_min} postulates that the minimum is always $O(\log n)$.

\paragraph{Organization of this paper:} In \cref{sec_prelims} we provide background  on Lewis weights, Laplacians and graph Effective Resistances, and prove our new characterization for Leverage Scores and Effective Resistance (Lemma \ref{lem_ER_char}). In \cref{sec_R_tot} we prove \Cref{thm_LW_apx_ERMP_trees,thm_UB_general_graphs}. \cref{sec_experiments} summarizes our experimental results. In \cref{sec_LW_spectral}, we prove \Cref{thm_spectral_LW} and the application for spectral-thin trees. Finally, in \cref{sec_computing_LW} we restate the algorithm for computing LW, and show that low-accuracy LW are sufficient for our results. We finish in \cref{sec_optimal_design} with exploring the relation of A- and D-optimal design, showing the geometric interpretation in terms of Pythagorean means, and proving \Cref{thm_natural_gap}.


\section{Preliminaries}\label{sec_prelims}

\subsection{Notations}
We denote by $\boldsymbol{S}^n_+$ the symmetric matrices subspace of $\R^{n \times n}$.
The unit vectors, $e_i$, are the vectors with all entries $0$, but the $i$'th entry (which equals  to $1$). 
We denote the eigenvalues (EV) of a matrix $M \in \R^{n \times n}$ by 
\[
\lambda(M) = \lambda_1(M) \leq \lambda_2(M) \leq \dots \leq \lambda_n(M).
\]


For any matrix $A \in \R^{n \times d}$, we denote the \emph{projection} matrix onto the column space of $A$ as
\[
\Pi_A = A(A^T A)^+ A^T.
\]

For future purposes, we restate here the AM-GM inequality. For any sequence $\{x_i\}_{i=1}^n$ of $n$ positive numbers, define:
\[
AM(X) = \frac{1}{n}\sum_{i} x_i \ , GM(X) = \sqrt[n]{x_1 \dots x_n} \ , \ HM(X) = \frac{n}{\sum_{i=1}^n x_i^{-1}} \ ,
\]
with the following inequality:
\[
     HM(X) \leq GM(X) \leq AM(X),
\]
when equality occurs iff $X$ is a constant sequence.



We use the following definition of ellipsoids - an ellipsoid $\mathcal{E}$ is defined by $\{ v \in \R^n \ \mid v^T A v \leq 1  \}$, where $A$ is a symmetric PSD matrix. The semiaxis of an ellipsoid $\mathcal{E}$ are given by $\boldsymbol{\sigma} = \{ \sigma_i(\mathcal{E}) \}_{i=1}^n$, and are equals to
\begin{align*}
\sigma_i = \lambda_i(A)^{-1/2} . \numberthis \label{equ_ellipsoid_semiaxis_def}
\end{align*}


From this it's clear that the volume of an ellipsoid equals , up to constant factor, to -- $\text{vol(}\mathcal{E}) \propto |A|^{-1/2}$.

\subsection{Leverage Scores and Lewis Weights}\label{sec_preliminary_LS}

For matrix $A \in \R^{n \times d}$ we define the \emph{Leverage Score} (LS) of the $i$'th row of $A$ as 
\[
\tau_i(A) = a_i^T (A^T A)^+ a_i , 
\]
where $a_i$ is the $i$'th row of $A$. We call $\tau(A) = \{\tau_i(A)\}_{i=1}^n$ the Leverage Scores of $A$.  Note that $\tau(A)$ is exactly the diagonal of the projection matrix $\Pi_A$. Since \(0 \preceq \Pi_A \preceq I\), we have $0 \leq \tau_i(A) \leq 1$, and :
\begin{fact}(Foster's lemma, \cite{Fos53}) \label{equ_LS_facts}
$\sum\limits_{i} \tau_i(A) = \Tr \ \Pi_A = \rank(A)$.
\end{fact}

Leverage scores are traditionally used as a tool for sketching and $\ell_2$-subspace embedding \cite{cw13}. Indeed, row-sampling of a matrix based on LS is known to produce a good spectral approximation \cite{cw13,dmmw12}. This is also the core idea behind spectral graph sparsification \cite{SS11}. \\

As mentioned in the introduction, we are interested in $\linf$-LW, which are a generalization of leverage scores:
\begin{definition}($\linf$-LW)\label{equ_inf_LW}
     For a matrix $A \in \R^{n \times d}$, the $\ell_\infty$-LW of $A$ ,denoted by $w_\infty(A) \in \R^n$, is the \emph{unique} weight vector such that for each row $i \in [n]$ we have
    \[
    (w_\infty)_i = \tau_i(W_\infty^{1/2}A) ,
    \]
    or equivalently,
    \[
        a_i^T(A^T W_\infty A)^+ a_i = 1  ,
    \]
    where $W_\infty = \diag(w_\infty)$. From now on, we will use LW to denote the $\linf$-LW$(A) \coloneqq w_\infty(A)$. 
\end{definition}

Note that Definition \eqref{equ_inf_LW} is \emph{cyclic} and indeed it is not a-priori clear that $\linf$-LW even exist. One way to prove the unique existence is by repeatedly computing the LS of the resulting matrix, until a fixed-point is reached, see \cref{sec_computing_LW}.  
Observe that for any matrix $A$, employing facts~\eqref{equ_LS_facts}, we have $w_\infty(A) \leq \buni$ and, because $W_\infty$ is a full diagonal matrix, meaning that $\rank(W_\infty^{1/2}A) = \rank(A)$,  $\sum_{i=1}^m (w_\infty(A))_i = \rank(A)$.

\subsection{Effective Resistance of Graphs}\label{sec_preliminary_graphs}

In this section we formally define the Effective Resistance (ER) of a graph, and derive a few key properties we will use later on. More details can be found in \cite{saberi}. 
Let $G=\langle V,E \rangle$, be an undirected graph with $n$ vertices and $m$ edges. w.l.o.g, $V = [n]$. We always assume that $G$ is connected and without self loops (if $G$ is not connected we can operate separately on each connected component). The degree of $u\in V$ is denoted $d(u) \coloneqq d_G(u)$. The adjacency matrix of $G$ is $A = A_G$. For each edge $l=(i,j)$ such that $i<j$, define an (arbitrary) \emph{orientation} $b_l \in \R^n$ as $(b_l)_i = 1$, $(b_l)_j = -1$ and all other entries $0$. $B = \{b_l\}_{l=1}^m \in \R^{n \times m}$ is called the \emph{edge-incidence} matrix of $G$. 




This paper explores the effects of \emph{re-weighting} a graph's edges on various graph functions. The re-weighted graph, $G=\langle V,E,w \rangle$, is defined by the associated weight vector $g \in \R^m$ such that, $g_l = w(e_l)$.

The corresponding \emph{weighted Laplacian} $L_g = L_g(G)$, w.r.t the weight vector $g \in \R^m$, is the PSD matrix 
\[
L_g := B \cdot \diag(g) \cdot B^T.
\]
Where clear from context, we sometimes denote $W = W_g \coloneqq \diag(g)$. 
Equivalently, 
$L_g = D_g - A_G(g)$, 
where $A_G(g)$ is the weighted adjacency matrix of $G$ and $D_g$ is diagonal matrix with $d_g(u) = \sum_{v \sim u} w(u,v)$ in the $u$'th diagonal term. Note that the unweighted case corresponds to $g= \buni\in \R^m$.\\

An important fact about Laplacians is that the rank of $L$ is $n-1$ and that $L \buni = 0$. It easily verified that $L +(1/n)\buni \buni^T$ is invertible and its inverse equals to,
\begin{equation}\label{equ_laplac_inv}
    (L + (1/n)\buni \buni^T)^{-1} = L^+ + (1/n)\buni \buni^T,
\end{equation}
where $L^+$ is the Pseudo-inverse of $L$.

Throughout the paper we will always assume that $g$ is normalized such that $\boldsymbol{1}^T g =1$. One nice consequence of this normalization is that $AM(\lambda(L_g)) = \frac{2}{n-1}$. To see why recall that $L_g = D_g - A(g)$. So
\[
    \Tr L_g = \Tr D_g - \Tr A(g) = \Tr D_g = \sum_{i} \sum_{j \sim i} w(i,j) =2 \cdot  \sum_{e \in E} w(e) = 2 \cdot \buni^T g = 2.
\]
$L_g$ has only $n-1$ positive EVs, so
\begin{align*}
 AM(\lambda(L_g)) = \frac{1}{n-1} \sum_{i=2}^n \lambda_i = \frac{1}{n-1} \Tr L_g = \frac{2}{n-1}. \numberthis \label{equ_AM_Laplacian}
\end{align*}
We will use that later on.\\

Throughout the paper, we shall use the following shorthand for the $\linf$-Lewis Weight of a given graph $G$  
\begin{equation}\label{equ_LW_normalize_grpahs}
    g_{lw} := \frac{1}{n-1} w_\infty(B^T) . 
\end{equation}
This normalization is due to fact~\eqref{equ_LS_facts} and that $\rank(B^T) = n-1$.\\

With these definitions, we can now formally define the Effective Resistances of a graph:
\begin{definition}[Effective Resistance] \label{def_ER}
Given a weighted graph $G$, with Laplacian $L_g = BWB^T$, the Effective Resistance (ER) between a pair of terminal nodes $(i,j)$ in the graph is: 
\[ ER(i,j) = b_{ij}^T L_g^+ b_{ij} \]
where, $b_{ij} = e_i - e_j$, and $L^+$ is Pseudo-inverse of $L$.
\end{definition}

Through the paper we sometimes denote the ER between two vertices $i$ and $j$ , with weight vector $g$ as $R_{ij}(g)$. We denote the Effective Resistance on the $l$'th edge, for $G(g)$ by
\[
R_l(g) = b_l^T L_g^+ b_l \ \ , \ l=1\dots m
\]
when the weights $g$ are clear, we will write $R_l$.
One useful property, which directly follows from the definition and properties of Pesudo-inverse is that ER is a homogeneous function of $g$ of degree $-1$, i.e
\begin{equation}\label{equ_ER_homogenous}
    R_{ij}(\alpha g) = R_{ij}(g)/\alpha.
\end{equation} 
Another important property of ER is that it defines a metric on the graph \cite{Klein93}, and as such satisfies the triangle inequality (for any weights $g$):
\begin{align*}
    R_{ij} \leq R_{ik} + R_{kj} \ \forall i,k,j \in V \numberthis \label{equ_ER_tri_ineq}
\end{align*}


Our paper continues the work of \cite{saberi} on the problem of minimizing the total Effective Resistance of a graph (ERMP), also known as the Kirchhoff Index $\cK_G(g)$  of $G$  \cite{Lukovits1999}. $\cK_G$ can also be expressed as
\begin{align*}
     \cK_G(g) := & \sum_{i,j} R_{ij}(g) \\
          = & n\Tr L_g^{+}  \numberthis \label{equ_R_tot_rep_tr} \\
          = & n\Tr (L_g + (1/n) \boldsymbol{1}\boldsymbol{1}^T)^{-1} - n 
\end{align*}
where we used the definition and trace laws for the first equality, and equation~\eqref{equ_laplac_inv} for the second one. 
Using the above expression, we can re-write the ERMP as the following SDP:
\begin{equation*}
\begin{array}{ll@{}ll}
\text{minimize}  & \Tr X^{-1} &\\
\text{subject to}&  X = \sum_{l=1}^{m} g_lb_lb_l^T + (1/n)\boldsymbol{1}\boldsymbol{1}^T \tag{ERMP}\\  \label{equ_R_tot_A_design}
                    & \boldsymbol{1}^Tg=1
\end{array}
\end{equation*}
which is a special case of A-optimal design (see~\eqref{equ_A_optimal_def}) \emph{over Laplacians} (with the incidence matrix $B$ as the experiment matrix). 


\subsection{A New Characterization of Leverage Scores and Effective Resistances}\label{sec_new_char_LS_ER}


Before proceeding to prove our results, we present here a new geometric characterization for leverage scores (LS) and effective resistances (ER). As mentioned earlier in the introduction, the optimal solution for D-optimal design is precisely the LW of the experiment matrix \cite{CP14,ccly}. Our characterization follows from a different (and to the best of our knowledge, a new) 
proof of the latter result, using elementary convex optimization analysis. For brevity, we only provide here a high-level overview, and defer the details to ~\cref{appendix_D_design}.  
Let $B$ be the edge incident matrix of some graph $G$, $L_g$ the weighted Laplacian of $G(g)$, and denote our  
objective (see remark~\ref{remark_LD_G_objective}) by 
\[ LD_G(g) := \logdet (L_g^+ + (1/n) \buni \buni^T).\] 
In~\eqref{equ_logdet_der}, we prove that 
the ER on the $l$'th edge equals  to
\begin{align*}
 R_l(g) = b_l^T L_g^+ b_l = b_l^T(B W_g B^T)^+ b_l = -\frac{\partial LD(g)}{\partial g_l}. 
\end{align*}
In other words, the gradient of $LD_G(g)$ (w.r.t $g$) equals  to the ER on the edges (up to a minus sign). 
This can be generalized to Leverage Scores of an arbitrary matrix as well: Let $V \in \R^{n \times m}$ and define $C = W_g^{1/2} V^T$, the weighted experiment matrix. Define $E_V(g) = V W_g V^T$ and $LD_V(g) = \logdet (E_V(g)^{-1}) = -\logdet (E_V(g))$. Then the $l$'th LS of $C$ equals  to
\begin{align}\label{eq_char_LS_grad}
&& \tau_l(C) = -g_l \cdot \frac{\partial LD_V(g)}{\partial g_l} && \text{see~\eqref{equ_der_LD_g_with_LS}}
\end{align}
I.e., the Leverage Scores of the weighted matrix are the weighted gradient of $LD_V(g)$ (up to a minus sign).

This characterization has the following geometric interpretation: We know that D-optimal design is equivalent to finding the minimum-volume enclosing ellipsoid on $V$. This means that the gradient of $LD_V(g)$ is proportional to the gradient of the volume of $\mathcal{E}$. Thus, \eqref{eq_char_LS_grad} implies:
\begin{lemma}
    The Leverage Scores of a matrix $V$ are equal (up to constant factors) to the gradient of the volume of the ellipsoid induced by $E_V(g)$.
    In particular, in the case of \emph{Laplacians}, the Effective Resistance on the edges of a weighted graph are equal to the gradient of the volume of the ellipsoid induced by $L_g$.
\end{lemma}

In fact, for Laplacians a stronger statement holds, namely, the optimality criterion for $LD_G$ becomes   
\begin{flalign*}
&& -\frac{\partial LD_G(g)}{g_l} = R_l(g) \leq \rank(B) = n-1. && \text{see~\eqref{equ_LD_G_opt_crit}}
\end{flalign*}
Moreover, at $g=g_{\ell w}$, we saturate this condition (see equation~\eqref{equ_LD_LW_opt_saturation}), meaning that
\[
R_l(g_{\ell w}) = n-1 \ , \ \forall \ l=1\dots m
\]

Now, from the uniqueness of the optimal solution for $LD(g)$, we know that for any other normalized weight vector $g$ the optimality criteria doesn't hold. In other words, for any $\hat{g} \neq g_{\ell w} \in \R^m$, such that $\buni^T \hat{g} = 1$, there exist $l \in [m]$ such that
\[
R_l \left( \ \hat{g} \ \right) > n-1.
\]
We conclude that $\linf-$LW \emph{minimizes the maximal} ER over  edges, among all (normalized) weight-vectors.


\section{ 
ERMP Approximation via $\linf$-Lewis Weights   
}\label{sec_R_tot}

Recall the Kirchoff Index of a graph, $\cK_G^*$, is the optimal solution for the ERMP problem \eqref{equ_ERMP_def}. Throughout this section, we denote by $\cK_G^{\ell w}$ the value of ERMP at $g_{\ell w}$.  The approximation ratio obtained by $\linf$-LW is defined as 
\begin{equation}\label{def_alpha_AD}
    \alpha_{A,D} \coloneqq \frac{\cK_G^{\ell w}}{\cK_G^*}. 
\end{equation}
Recall that $g_{\ell w}$ is the optimal solution for the D-optimal design, so $\alpha_{A,D}$ exactly captures the gap of A- and D-optimal design over Laplacians. Since we are proposing $\linf$-LW as an approximation-algorithm for the ERMP, $\alpha_{A,D}$ is simply the \emph{approximation ratio} of our algorithm.
\\
This section is divided to roughly two independent results. We first focus solely on \emph{trees} thus proving \Cref{thm_LW_apx_ERMP_trees}. For trees, the ERMP gets a more simple form, and we used a designated technique for proving our theorem. The second result is an  attempt to expand to general graphs. We give a more general analysis, and showing two different upper bounds on our approximation ratio. We conjecture that they yield an $O(\log n)$-UB, and in particular implying \Cref{conj_LW_apx_ERMP}. We finish this section by showing some formulas regards this UBs.

\subsection{$\linf$-LW are $O(1)$-Approximation for Trees}

We divide this section to two parts. First, we rewrite the approximation ratio especially for trees, using the formulas for the optimal weights derived by \cite{saberi}. Then we give the technical details of the proof of \Cref{thm_LW_apx_ERMP_trees}.\\ 

\noindent We first show that for trees $g_{\ell w} = g_{uni} = (1/m)\buni$: 
\begin{claim}\label{clm_LW_tree}
    For any tree $T$, we have $g_{\ell w} = g_{uni}$.
\end{claim}
\begin{proof}
    For simplicity, let's look first at the non-normalized weights, $w_\infty = w_\infty(B^T)$. We know that $w_\infty \leq \buni$ (equation~\eqref{equ_LS_facts}). Moreover, we have by~\eqref{equ_LS_facts} that 
    \begin{align*}
        \sum_{l=1}^m (g_{\ell w})_l = \underbrace{n-1 = m}_\text{for trees}.
    \end{align*}
    Thus, having $m$ positive numbers, which at most $1$ and sums up to $m$, inevitably are all equal to $1$. Hence,
    $ (g_{\ell w})_l = 1 \ , \forall \ l=1 \dots m$, so normalizing by $n-1 = m$ we  get $g_{\ell w}=g_{uni}$, as claimed.

\end{proof}

Next, we restate some key formulas from \cite{saberi} for the ERMP on trees. First, when the graph is a tree, for any two vertices $i,j$, the ER between them is
\[
R_{ij}(g) = \sum_{e_l \in P(i,j)} \frac{1}{g_l}
\]
This is clear from an electrical network POV\footnote{i.e The ER of resistors in series is additive.}. In other words, the $l$'th edge contributes $(1/g_l)$ to all the paths that contain it. With this in mind, we define the \emph{congestion} of the $l$'th edge as follows:

\begin{definition}[congestion of an edge]\label{def_congestion}
The congestion of an edge $e_l \in T$, denoted $c_l(T)$, is the number of paths in $T$ that contain it.
\end{definition}
It's easy to see that 
\[
c_l =  n_l(n-n_l) ,
\]
where $n_l$ is the number of nodes in the \emph{sub-tree} of one side of $e_l$ (this is symmetric). Note that from the concavity of the function $f(x) = x(n-x)$, the minimal congestion is achieved at the leaves of $T$ ($1(n-1) = m$), whereas the maximal congestion is 
$\frac{n}{2}(n-\frac{n}{2}) = (n/2)^2$. 
More generally, 
\begin{align*}
f(x) = x(n-x) < y(n-y) = f(y) \iff \min(x,n-x) < \min(y,n-y). \numberthis \label{equ_cong_relation}
\end{align*}
With this definition and the above fact, we get that for any tree $T$, and for any weight-vector $g \in \R^m$, 
\[
\cK_T(g) = \sum_{l=1}^m \frac{c_l}{g_l} \;\; .
\]
In addition, \cite{saberi} (Section 5.1) proved that, for trees, the optimal ERMP solution is  
\begin{equation}\label{equ_opt_tree}
    \cK_T^* = \left(\sum_{l} c_l^{1/2} \right)^2, 
\end{equation}
and that the closed-form expression for the optimal weights is 
\begin{equation}\label{equ_opt_g_trees}
    g^*_l = (c_l/\cK^*)^{1/2} = \frac{c_l^{1/2}}{\sum_{k} c_k^{1/2}} \  \ \ \ \forall l=1,\dots,m.
\end{equation}
Using the last expressions, it immediately follows that 
\begin{equation}\label{equ_cong_sum}
    c_l = \cK^* \cdot (g^*_l)^2 \implies \sum_l c_l = \cK^* \cdot \sum_l (g^*_l)^2 =  \cK^* \cdot ||g^*||_2^2 . 
\end{equation}
Lastly, our proof will use the fact that the optimal weights are at most $1/2$. Indeed,
\begin{align*}
g^*_{max} \leq \frac{c_{max}^{1/2}}{\min \{\sum_l c_l^{1/2} \}} \leq \frac{n/2}{m\sqrt{m}} \leq \frac{1}{\sqrt{m}} \ll \frac{1}{2}, \numberthis \label{equ_g_star_UB}
\end{align*}
when we used the bounds from \eqref{equ_cong_relation}.\\

Using the fact that $g_{\ell w}= g_{uni}$ (claim~\ref{clm_LW_tree}), we can write $\cK^{\ell w}_T$ as:
\[
\cK_T^{\ell w} =\sum_{l=1}^m \frac{c_l}{(g_{\ell w})_l} = \sum_{l=1}^m \frac{c_l}{1/m} = m\sum_{l} c_l  
\]
Thus, for any tree $T_n$ of order $n$, we can write the approximation ratio in terms of the congestion as:
\begin{equation}\label{equ_apx_ratio_trees}
    \alpha_{A,D}(T_n) = \alpha(T_n) = \frac{m\sum\limits_{l} c_l}{\left(\sum\limits_{l} c_l^{1/2} \right)^2} 
\end{equation}

We shall also use the following claim 
in the proof of Theorem \ref{thm_LW_apx_ERMP_trees} :
\begin{claim}\label{clm_apx_ratio_subtree}
    Let $T,T'$ be two trees such that $T \subset T'$. Then $\alpha(T) \leq \alpha(T')$.
\end{claim}
\begin{proof}
    Recall the definition of $\alpha(T)$:
    \[
        \alpha(T) = \frac{\cK^{\ell w}_T}{\cK^*_T} = \frac{m \cdot \sum_l c_l}{\cK^*_T}
    \]
    Clearly, $\cK^{\ell w}_T < \cK^{\ell w}_{T'}$ (we are only adding positive terms). Moreover, \cite{saberi} showed that 
    \[
    T \subset T' \implies \cK^*_{T'} \leq \cK^*_T \;\; .
    \]
    (this can be seen as the optimal weight for $T$ is feasible for $T'$). From this, it's easy to see that
    \[
        \alpha(T) \leq \alpha(T')
    \]
\end{proof}






\subsubsection{Proof of Theorem \ref{thm_LW_apx_ERMP_trees}}\label{sec_prf_LW_apx_trees}

We begin by showing when an LT increases the approximation ratio. 
\begin{proof}[Proof of Claim~\ref{clm_apx_ratio_ET}]
\normalfont
    Denote by $T' = E_k \circ T$, $c_k' = c_k(T')$ (the rest of the congestions remains the same).
    Recall that $\alpha = (m\sum_l c_l)/ \left(\sum_l c_l^{1/2}\right)^2$, so let's write the modified numerator and denominator. 
    \begin{align*}
        &\sum_{l} c_l' = \sum_{l} c_l + (c_k' - c_k). \\
        &\sum_{l} c_l'^{1/2} = \sum_{l} c_l^{1/2} +(c_k'^{1/2} - c_k^{1/2}).
    \end{align*}
    
    We get that,
    \begin{align*}
        \alpha(T) \leq \alpha(T') &\iff \frac{\sum_{l} c_l}{\left(\sum_l c_l^{1/2}\right)^2} \leq \frac{\sum_{l} c_l + (c_k' - c_k)}{\left(\sum_{l} c_l^{1/2} +(c_k'^{1/2} - c_k^{1/2})\right)^2} \\
        &\iff \frac{\sum_{l} c_l}{\left(\sum_l c_l^{1/2}\right)^2} \leq \frac{\sum_{l} c_l + (c_k' - c_k)}{\left(\sum_{l} c_l^{1/2}\right)^2 +\left(c_k'^{1/2} - c_k^{1/2}\right)^2 + 2\left(\sum_{l} c_l^{1/2}\right) \left(c_k'^{1/2} - c_k^{1/2}\right) } \\
    \end{align*}
    
    Using the fact that, assuming the denominators are positive (which true in our case),
    \[
    \frac{A}{B} \leq \frac{A+D}{B+C} \iff AC \leq BD,
    \]
    we get that
    \begin{align*}
        \alpha(T) \leq \alpha(T') &\iff \left(\sum_l c_l \right)\left[ \left(c_k'^{1/2} - c_k^{1/2}\right)^2 + 2\left(\sum_{l} c_l^{1/2}\right) \left(c_k'^{1/2} - c_k^{1/2}\right) \right]  \leq (c_k' - c_k)\left(\sum_{l} c_l^{1/2}\right)^2 \\
        &\iff \left(\sum_l c_l \right) \left(c_k'^{1/2} - c_k^{1/2}\right) \left( 2\sum_{l} c_l^{1/2} + c_k'^{1/2} - c_k^{1/2} \right)  \leq (c_k'^{1/2} - c_k^{1/2}) (c_k'^{1/2} + c_k^{1/2})\left(\sum_{l} c_l^{1/2}\right)^2 \\
    \end{align*}

    We divide the rest of the proof to two cases.\\
    
    \textbf{Case 1:} $E_k$ is an upper LT. By definition, $c_k' > c_k$. Then, $c_k'^{1/2} - c_k^{1/2} > 0$, so we can divide by the latter to get, 
    \begin{align*}
        \alpha(T) \leq \alpha(T') &\iff \left(\sum_l c_l \right) \left( 2\sum_{l} c_l^{1/2} + c_k'^{1/2} - c_k^{1/2} \right)  \leq (c_k'^{1/2} + c_k^{1/2})\left(\sum_{l} c_l^{1/2}\right)^2 \\
        &\iff ||g^*||_2^2 \left( 2\frac{c_k^{1/2}}{g^*_k} + c_k'^{1/2} - c_k^{1/2} \right) \leq (c_k'^{1/2} + c_k^{1/2}) && \text{using~\eqref{equ_cong_sum}, and \eqref{equ_opt_g_trees} for the $k$'th edge.}\\
        &\iff  2\frac{c_k^{1/2}}{g^*_k} - c_k^{1/2}\left(1 + \frac{1}{||g^*||_2^2} \right) \leq c_k'^{1/2}\left(\frac{1}{||g^*||_2^2} -1\right) && \text{group terms by $c_k^{1/2}$ and $c_k'^{1/2}$.}
    \end{align*}
   
    In our case, $c_k'/c_k > 1$, so after dividing both sides by $c_k^{1/2}$, it's sufficient to require
    \begin{align*}
        \frac{2}{g^*_k} - 1 - \frac{1}{||g^*||_2^2}  < \frac{1}{||g^*||_2^2} -1 \iff g_k^* > ||g^*||_2^2 .
    \end{align*}
    Thus we get that,
    \[
     g_k^* > ||g^*||_2^2 \implies \alpha(T) \leq \alpha(T').
    \]

    \textbf{Case 2:} $E_k$ is a lower LT . Using the same arguments we get that,
    \[
    g_k^* \leq ||g^*||_2^2 \implies \alpha(T) \leq \alpha(T')
    \]
    Unifying the two cases completes the proof. 
\end{proof}

Remember that we defined the following partition:
\[
E_<(T) \coloneqq \{ l \ \mid \ g_l^*(T) \leq ||g^*(T)||_2^2 \} \ , \ E_>(T) \coloneqq \{ l \ \mid \ g_l^*(T) > ||g^*(T)||_2^2 \}
\]

We now show the key feature of this partition -- $E_<$ and $E_>$ are \emph{invariant} under upper and lower LT:
\begin{claim}\label{clm_invariant_LT}
    For any tree $T$, and $k \in E_>$ we have that,
    \[
    E_>(E_k^{\uparrow} \circ T) = E_>(T).
    \]
    and vice versa for $E_<$.
\end{claim}
\begin{proof}
\normalfont
    We will only prove this for $E_>$ as the proof for $E_<$ is symmetric. Let $k \in E_>(T)$. Denote by $T'=E_k^{\uparrow} \circ T$. We first show that $g^*_k(T') > g^*_k(T)$. From equation~\eqref{equ_opt_g_trees} we know that
    \[
    g_l^*(T) = \left(\frac{c_l(T)}{\cK^*_T}\right)^{1/2} \ , \ g_l^*(T') = \left(\frac{c_l(T')}{\cK^*_{T'}}\right)^{1/2}
    \]
    In addition, from the definition of upper LT, the congestion on the $k$'th edge increases, and doesn't change on any other edge, meaning that:
    \[
    \cK^*_{T'} = \left(\sum_l c_l(T')^{1/2} \right)^2 > \left(\sum_l c_l(T)^{1/2} \right)^2 = \cK^*_T.
    \]
    
    Combining the last two equations ,implies that for all $l \neq k$ we have:
    \begin{equation}\label{equ_g_l_upper_LT}
            c_l(T') = c_l(T) \implies \frac{g_l^*(T)}{g_l^*(T')} = \frac{\cK^*_{T'}}{\cK^*_T} > 1 \implies g_l^*(T') < g_l^*(T).
    \end{equation}
    
    Recall that $\buni^T g^* = 1$, so,
    \begin{equation}\label{equ_g_k_upper_LT}
        g_k^*(T') = 1-\sum_{l\neq k} g^*_l(T') > 1 - \sum_{l\neq k} g^*_l(T) = g_k^*(T). 
    \end{equation}
    
    We need to show that $k \in E_>(T')$, i.e
    \[
    g_k^*(T') > ||g^*(T')||_2^2
    \]
    
    Using~\eqref{equ_g_k_upper_LT}, we can write:
    \begin{align*}
        g_k^*(T') &= g_k^*(T) + (g_k^*(T') - g_k^*(T)) \\
        &> \sum_{l \neq k} g_l^*(T)^2 + g_k^*(T)^2 +(g_k^*(T') - g_k^*(T)) && \text{since $k \in E_>(T)$} \\
        &> \sum_{l \neq k} g_l^*(T')^2 + g_k^*(T)^2 +(g_k^*(T') - g_k^*(T)) && \text{using~\eqref{equ_g_l_upper_LT}} \\
        &= \sum_{l \neq k} g_l^*(T')^2 + g_k^*(T')^2 + \left(g_k^*(T') - g_k^*(T) + g_k^*(T)^2 - g_k^*(T')^2 \right) \\
        &> \sum_{l} g_l^*(T')^2 = ||g^*(T')||_2^2
    \end{align*}
    The last inequality is because $\forall l , \ g^*_l < 1/2$ (see~\eqref{equ_g_star_UB} above), and for any two number $a,b$ such that $a-b>0 , a+b<1$ we have that:
    \[
    a-b > (a-b)(a+b) = a^2 - b^2 \implies (a-b)-(a^2-b^2) >0.
    \]
    Using this fact with $a=g_k^*(T') \ , \ b=g_k^*(T)$ concludes the proof.
\end{proof}

So far we showed that if there exist an LT as defined in~\eqref{equ_upper_lower_LT} then the sets $E_<, \ E_>$ are invariant under the matching transformations, and the approximation ratio of the transformed tree is larger than the original tree. Using these properties, we can repeatedly apply an LT (until saturation), and the final tree is guaranteed to be the hardest instance. So, we are ready to describe the main process of this section. We begin with arbitrary tree $T$ of order $n$. We partition its edges to two sets $E_<(T),E_>(T)$ as defined above. We know that this two sets are invariant under lower and upper LT - $E^{\uparrow} , E^{\downarrow}$ - so we can repeatedly apply them on both sets until we have $E \circ T = T$. Furthermore, from claim~\ref{clm_apx_ratio_ET} we can say that for the final tree $\Tilde{T}$, $\alpha(\Tilde{T})$, is an UB for the approximation ratio of the original tree. It is left to explicitly define $E^{\uparrow} , E^{\downarrow}$ and compute $\alpha(\Tilde{T})$.
\\
\begin{figure}
    \centering
    \begin{tikzpicture}
	\begin{pgfonlayer}{nodelayer}
		\node [style=none] (2) at (-5.75, 1) {};
		\node [style=none] (3) at (-4, -0.75) {};
		\node [style=none] (4) at (-6.4, -0.7) {};
		\node [style=none] (5) at (-8, -0.75) {};
		\node [style={tree_node}] (8) at (-5.725, 3.65) {$T_0$};
		\node [style=none] (9) at (-5.425, 1.15) {$v$};
		\node [style=none] (10) at (-5.5, 2.625) {$u$};
		\node [style=none] (11) at (-8.35, -0.5) {$x_1$};
		\node [style=none] (12) at (-6.75, -0.5) {$x_2$};
		\node [style=none] (13) at (-6.075, -1.375) {};
		\node [style=none] (14) at (-4.275, -1.375) {};
		\node [style=none] (15) at (-4.5, -2) {};
		\node [style=none] (16) at (-3.5, -2) {};
		\node [style=none] (17) at (-5.9, -2) {};
		\node [style=none] (18) at (-6.875, -2) {};
		\node [style=none] (19) at (-7.5, -2) {};
		\node [style=none] (20) at (-8.5, -2) {};
		\node [style=none] (21) at (-3.475, -0.575) {$x_d$};
		\node [style=none] (22) at (-6.175, 1.85) {$e_k$};
		\node [style=none] (23) at (-1.25, 0.75) {$\longrightarrow$};
		\node [style={tree_node}] (24) at (4, 3.625) {$T_0$};
		\node [style=none] (26) at (6.5, 0.55) {};
		\node [style=none] (28) at (6.75, 0.375) {$v$};
		\node [style=none] (30) at (4.65, 2.725) {$u$};
		\node [style=none] (31) at (6, 1.5) {$e_k$};
		\node [style=none] (32) at (4, 2.825) {};
		\node [style=none] (33) at (5.725, -0.65) {};
		\node [style=none] (34) at (3, -0.675) {};
		\node [style=none] (35) at (1.025, -0.675) {};
		\node [style=none] (36) at (0.5, -1.925) {};
		\node [style=none] (37) at (1.5, -1.925) {};
		\node [style=none] (38) at (2.5, -1.925) {};
		\node [style=none] (39) at (3.5, -1.925) {};
		\node [style=none] (40) at (5.25, -1.925) {};
		\node [style=none] (41) at (6.25, -1.925) {};
		\node [style=none] (42) at (3.425, -1.425) {};
		\node [style=none] (43) at (5.325, -1.425) {};
		\node [style=none] (44) at (0.625, -0.425) {$x_1$};
		\node [style=none] (45) at (2.7, -0.425) {$x_2$};
		\node [style=none] (46) at (5.2, -0.475) {$x_d$};
		\node [style={fill_node}] (47) at (4, 2.825) {};
		\node [style={fill_node}] (48) at (-5.75, 2.85) {};
		\node [style={fill_node}] (49) at (-5.75, 0.95) {};
		\node [style={fill_node}] (50) at (6.5, 0.525) {};
		\node [style=none] (51) at (-10.5, -0.5) {$T_v$};
		\node [style=none] (53) at (-9.25, -1.975) {};
		\node [style=none] (54) at (-9.25, 0.9) {};
	\end{pgfonlayer}
	\begin{pgfonlayer}{edgelayer}
		\draw (2.center) to (3.center);
		\draw (2.center) to (4.center);
		\draw (2.center) to (5.center);
		\draw [style={selected_edge}] (8) to (2.center);
		\draw [style={dash_lines}, in=180, out=0] (13.center) to (14.center);
		\draw (3.center) to (15.center);
		\draw (15.center) to (16.center);
		\draw (16.center) to (3.center);
		\draw (4.center) to (18.center);
		\draw (18.center) to (17.center);
		\draw (17.center) to (4.center);
		\draw (5.center) to (20.center);
		\draw (20.center) to (19.center);
		\draw (19.center) to (5.center);
		\draw (35.center) to (32.center);
		\draw (34.center) to (32.center);
		\draw (33.center) to (32.center);
		\draw (35.center) to (37.center);
		\draw (37.center) to (36.center);
		\draw (36.center) to (35.center);
		\draw (38.center) to (39.center);
		\draw (39.center) to (34.center);
		\draw (38.center) to (34.center);
		\draw (40.center) to (41.center);
		\draw (41.center) to (33.center);
		\draw (33.center) to (40.center);
		\draw [style={dash_lines}] (42.center) to (43.center);
		\draw [style={selected_edge}] (32.center) to (26.center);
		\draw [style={vertical_curly_braket}] (54.center) to (53.center);
	\end{pgfonlayer}
\end{tikzpicture}
    \caption{Description of lower LT }
    \label{fig_lower_ET}
\end{figure}
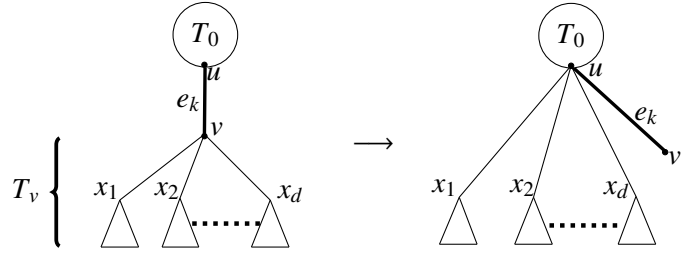

We first define $E^{\downarrow}$. We are assuming that the tree is rooted, and "up" and "down" directions are well defined. Let $k \in E_<(T)$ be a non leaf edge (meaning its congestion is greater than $m$). Let $e_k = (u,v)$. Denote by $T_v$ the downward sub tree rooted at $v$. Then $E^{\downarrow}_k$ is to take $T_v$ and "ascend" it to $u$ (see Figure~\ref{fig_lower_ET}). Clearly, for any $l \neq k$ , $c_l(E^{\downarrow}_k \circ T) = c_l(T)$ so $E^{\downarrow}_k$ is an LT. Furthermore $c_k(E^{\downarrow}_k \circ T) = m < c_k(T)$ 
 \eqref{equ_cong_relation}, so $E^{\downarrow}_k$ is indeed a lower LT. In other words, applying $E^{\downarrow}_k$ on all $k \in E_<(T)$ simply turns all of them into leafs in $T'$. So, we first apply $E^{\downarrow}_k$ on all edges $k \in E_<(T)$, leaving us with $T'$ such that for any $l \in E_<(T')$ , $e_l$ is a leaf edge. So from now on we will assume that all edges in $E_<(T)$ are leaf edges.  
\\

Next, we define $E^{\uparrow}$. We can focus on the sub-tree, $T_>$, spanned by $E_>(T)$ - all the other edges are leafs, so the structure of both trees is the same. Let $k \in E_>(T)$. Denote $e_k = (u,v)$.  If $d_{T_>}(u) , d_{T_>}(v) \leq 2 $, do nothing. We emphasize that $d(u)$ is the degree of $u$ in the \textbf{sub-tree} $T_>$ - i.e number of edges of $u$ that in $E_>$. We now divide the transformation to two cases: \\
\\
\textbf{Case 1:} $d_{T_>}(u) , d_{T_>}(v) > 2 $ (see Figure~\ref{fig_upper_ET_case1}). Let $n_v = |T_v| $. w.l.o.g. assume that $n_v \leq n/2$. Denote by $l_i = (v,x_i)$ for $i=1,...,d(v)$ such that $|T_{x_1}| = n_1 \leq n_2 \leq ... \leq n_{x_{d(v)}}$. Then, $E^{\uparrow}_k$ is to change $(v,x_2)$ to $(x_1,x_2)$.
Clearly the congestion changes only on $l_1'=(x_1,x_2)$. But, since $n_1,n_2 \leq n_1+n_2 \leq n_v \leq n/2$, then
\begin{align*}
    \min&(n_1,n-n_1) = n_1 < n_1+n_2 = \min(n_1+n_2,n-(n_1+n_2)) \\
    &\implies \underbrace{b_1(T) = n_1(n-n_1)}_{\text{congestion on $l_1$}} < \underbrace{(n_1+n_2)(n-(n_1+n_2)) = b_1(E^{\uparrow}_k \circ T)}_\text{congestion on $l_1'$}.
\end{align*}

So, indeed $E^{\uparrow}_k$ is an upper LT\footnote{While $E^{\uparrow}_k$ is indeed defined by $k$ it actually changes the congestion on $l_1$. While this is somewhat ambiguous, since $l_1 \in E_>$ it doesn't affect the proof, so we abuse this notation.}. \\
\\
\textbf{Case 2:} w.l.o.g.  $d(u)=2 , \ d(v) >  2$ (see Figure~\ref{fig_upper_ET_case2.1}).  
Similar to the first case, let $l' = (u,y) \ , \ l_i = (v,x_i)$ for $i=1,...,d(v)$ such that $|T_{x_1}|  = n_1 \leq n_2 \leq ... \leq n_{x_{d(v)}} \ , \ n' = n-|T_v|$.\\

\textbf{Case 2.1:} If $n_1 \leq n'$, then once again $E^{\uparrow}_k$ is to change $(v,x_2)$ to $(x_1,x_2)$.
The congestion changes only on $l_1'=(x_1,x_2)$, but $n_1 < n/(d(v)-1) \leq n/2$ and $n_1<n_i,n'$, so
\begin{align*}
    \min&(n_1,n-n_1) = n_1 < \min (n_1+n_2 , 1+n' + n_3+\dots)\\
    &\implies b_1(T) = n_1(n-n_1) < (n_1+n_2)(n-(n_1+n_2)) = b_1(E^{\uparrow}_k \circ T),
\end{align*}

and, indeed, $E^{\uparrow}_k$ is an upper LT.
\\

\textbf{Case 2.2:} If $n' < n_1$, then $E^{\uparrow}_k$ is to take $(v,x_1)$ with $T_{x_1}$ and move it to $(u,x_1)$ (see Figure~\ref{fig_upper_ET_case2.2}). The congestion changes only on $e_k=(u,v)$. But $n' < n/(d(v)-1) \leq n/2 \ , \ n' < n_1 \leq n_i$ so,
\begin{align*}
    \min&( n',n-n') = n' < \min ( n_1+n' , 1+n_2+\dots)\\
    &\implies c_k(T) = n'(n-n') <  (n'+n_1)(n-(n'+n_1)) = c_k(E^{\uparrow}_k \circ T),
\end{align*}
so, again, we get that $E^{\uparrow}_k$ is an upper LT.
\\

\begin{remark}
    Note that at the end of case \textbf{(2.2)} we either move to case \textbf{(1)} (if $d(v) > 3$) or to case \textbf{(2.1)} (if $d(v) =3$), so we can focus only on the first type of upper LT (first two cases). Now, for this type, the degree of $v$ after the transformation is strictly smaller than before, so this process will terminate at some point and we guarantee to end up with $d(v) = 2$ for all $v \in E_>(T)$ - i.e \emph{path graph} (of edges in $E_>$).
\end{remark}  

\begin{figure}[t]
    \centering
    \begin{tikzpicture}
	\begin{pgfonlayer}{nodelayer}
		\node [style={vertex_text_node}] (0) at (-6.25, 0.75) {\scriptsize $u_1$};
		\node [style={vertex_text_node}] (1) at (-3, 0.75) {\scriptsize $u_2$};
		\node [style={vertex_text_node}] (2) at (0, 0.75) {\scriptsize $u_3$};
		\node [style={vertex_text_node}] (3) at (4.25, 0.75) {\scriptsize $u_{p}$};
		\node [style={fill_node}] (4) at (-7.5, 2.25) {};
		\node [style={fill_node}] (5) at (-7.75, -0.75) {};
		\node [style={fill_node}] (6) at (-3.75, 2.25) {};
		\node [style={fill_node}] (7) at (-3.75, -0.75) {};
		\node [style={fill_node}] (8) at (-2.75, -0.75) {};
		\node [style={fill_node}] (9) at (0.75, -0.75) {};
		\node [style={fill_node}] (10) at (5, -0.75) {};
		\node [style={fill_node}] (11) at (4.75, 2.25) {};
	\end{pgfonlayer}
	\begin{pgfonlayer}{edgelayer}
		\draw [style={selected_edge}] (0) to (1);
		\draw [style={selected_edge}] (1) to (2);
		\draw [style={dash_thin}] (2) to (3);
		\draw [style={blue_edge}] (4) to (0);
		\draw [style={blue_edge}] (0) to (5);
		\draw [style={blue_edge}] (6) to (1);
		\draw [style={blue_edge}] (1) to (8);
		\draw [style={blue_edge}] (1) to (7);
		\draw [style={blue_edge}] (2) to (9);
		\draw [style={blue_edge}] (11) to (3);
		\draw [style={blue_edge}] (3) to (10);
	\end{pgfonlayer}
\end{tikzpicture}
    \caption{Stem graph with branches. The "black dots" are leafs, the black edges are from $E_>$, and the blue edges from $E_<$.}
    \label{fig_path_exits}
\end{figure}
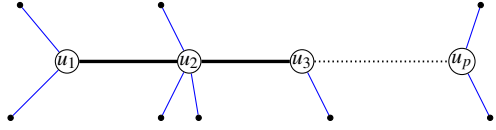

We can finally describe the final tree $\Tilde{T}$. As mentioned earlier, we apply lower LT until all edges from $E_<$ are leafs. In addition, we apply upper LT until we get a single path of edges from $E_>$. Thus the final tree has a "stem" structure (i.e path) of size $p=|E_>(T)|$ with $s=|E_<(T)|$ "branches" along the way (see \Cref{fig_path_exits}). At this point it should be clear that "pushing aside" branches to the closest side would increase the approximation ratio, since this increases the bottleneck in the tree. Formally speaking, pushing leafs to the end will increase the congestion along the path which is again an upper LT, thus increasing the approximation ratio as well (remember that edges on the path are from $E_>$). So the final tree $\Tilde{T}$ is simply
\[
\Tilde{T} = \mathcal{B}_{s_1,p,s_2}
\]
where $s_1+s_2 = s$ , and $\mathcal{B}_{t,p,s}$ is the \emph{Bowtie} graph (see \Cref{fig_tps_tree}).\\

All that's left is to compute $\alpha(\mathcal{B}_{t,p,s})$. We know that $T \subset T' \implies \alpha(T) \leq \alpha(T')$ (see claim~\ref{clm_apx_ratio_subtree}), so it's sufficient to prove for $\mathcal{B}_{n,n,n}$ ($n$ is the size of the original tree). While we can do it explicitly, we will use a much easier calculation using an UB we derive ahead for general graphs, giving us an UB of $\sim 3$, and concluding the proof (see lemma~\ref{lem_a2_bowtie} for further details).

\subsection{ERMP for General Graphs}\label{sec_general_UB}

We now continue to the case of general graphs. Unlike for trees, for general graphs the ERMP doesn't have a clear "combinatorical" form but rather an algebraic expression. Hence, our proving method should take a different aim. We derive two different UB for the approximation ratio, and we conjecture that their \emph{minimum} is $\Tilde{O}(1)$. Indeed, we will show that in the case of Bowtie graph the \emph{minimum} is $\Tilde{O}(1)$, matching \Cref{thm_LW_apx_ERMP_trees}. Finally, we provide many simulations that support our conjecture.  \\


Recall that the ERMP can be formulated as:
\begin{equation*}
\begin{array}{ll@{}ll}
\text{minimize}  & \Tr X^{-1} &\\
\text{subject to}&  X = \sum_{l=1}^{m} g_lb_lb_l^T + (1/n)\boldsymbol{1}\boldsymbol{1}^T \\
                    & \boldsymbol{1}^Tg=1
\end{array}
\end{equation*}





Since a closed-form expression for the exact approximation ratio seems hard to compute, 
our goal is to derive an upper bound on the approximation ratio of $\linf$-LW, which is henceforth denoted $\alpha_{A,D}$. It turns out that we can derive two \emph{different} upper bounds (with interesting relation) for $\alpha_{A,D}$. The first bound is using the $AM-GM$ inequality, and the second bound is via Lagrange duality analysis.

\subsubsection{The AM-GM UB}
Recall that we can formulate both A-optimality and D-optimality in terms of "means" (see \cref{sec_optimal_design}). So trying to use $AM-GM$ inequality makes a lot of sense. We know we can write $\cK_G(g)$ as (equation~\eqref{equ_R_tot_rep_tr})
\begin{equation}\label{equ_R_tot_rep_HM}
    \cK(g) = n \Tr L_g^+ = n \sum_{i=1}^{n-1} \lambda_i(L_g^+) = n  \sum_{i=1}^{n-1} \lambda_i(L_g)^{-1} = \frac{n(n-1)}{HM(\lambda(L_g))}
\end{equation}
(we only have $n-1$ positive eigenvalues). So, we get
\[
HM(\lambda(L_{g^*})) = \frac{n(n-1)}{\cK_G^*} \ ,\  HM(\lambda(L_{g_{\ell w}})) = \frac{n(n-1)}{\cK_G^{\ell w}}
\]
Using \emph{AM-GM} inequality, with equation~\eqref{equ_AM_Laplacian}, gives us,
\[
\alpha_{A,D} = \frac{\cK_G^{\ell w}}{\cK_G^*} = \frac{HM(\lambda(L_{g^*}))}{HM(\lambda(L_{g_{\ell w}}))} \leq \frac{AM(\lambda(L_{g^*}))}{HM(\lambda(L_{g_{\ell w}}))} = \frac{2/(n-1)}{HM(\lambda(L_{g_{\ell w}}))} = \frac{2}{n(n-1)^2} \cK_G^{\ell w}
\]

We denote this bound, the \emph{AM-GM} bound, by $\alpha_{AM}$:
\[
\alpha_{A,D} \leq \alpha_{AM}(\ell w) = \frac{2}{n(n-1)^2} \cK_G^{\ell w} = \frac{2}{(n-1)^2} \Tr L_{\ell w}^+
\]

It is nice to see that, since $\cK_G^* \geq n(n-1)^2/2$ (see lower bound on the optimal solution from \cite{saberi}) we get that:
\[
\alpha_1(g) = \frac{2}{n(n-1)^2}\cK_G(g) \geq \frac{2}{n(n-1)^2}\cK_G^* \geq 1 ,
\]
with equality iff $g=g^*$, so this is well defined.

\subsubsection{The Duality-gap UB}
Next we would like to use the \emph{dual gap} of ERMP to derive a second UB. The usage of the dual problem to bound the sub-optimality of a feasible solution is a well known technique, so our motivation is clear. We begin by restating the result of \cite{saberi} regards the duality gap of the ERMP problem. For the sake of complicity as it is one of the key component in our result, we derive the duality gap from scratch, but note that this is already done by \cite{saberi}. We state here the final results, see the full proof in \cref{appendix_ERMP_dual_gap}.\\

The dual problem of ERMP is (see equation~\eqref{equ_ERMP_dual_prob}),
\begin{equation*}
    \begin{array}{ll@{}ll}
     & \text{maximize}  \ \ & n(\Tr V)^2  \\
     & \text{subject to} &  b_l^T V^2 b_l = ||V b_l|| \leq 1 , \\ 
     & & V \succeq 0 \ ,\  V \boldsymbol{1} = 0
\end{array}
\end{equation*} 

We know that the gradient of $\cK_G(g)$, w.r.t to $g$, is equals to (see equation~\eqref{equ_der_L_g_plus}):
\[
\frac{\partial L_g^+}{\partial g_l} = -\Tr L_g^+ b_l b_l^T L_g^+ = -||L_g^+ b_l||^2
\]
Given any feasible $g$, let $V = \frac{1}{\underset{l}{\max} ||L_g^+ b_l||} L_g^+$ (this is clearly dual feasible). By definition, the duality gap in $g$, is equals to,
\begin{align*}
    \eta &= n \Tr L_g^+ - n(\Tr V)^2 = n \Tr L_g^+ - \frac{n (\Tr L_g^+)^2}{\underset{l}{\max}  ||L_g^+ b_l||^2} \\
    &= \frac{\cK_G(g)}{-\min\limits_l \ \partial \cK_G(g)/\partial g_l} (-\min_l \frac{\partial \cK_G(g)}{\partial g_l} - \cK_G(g))
\end{align*}
Using that, we can derive our second UB. By the definition of dual gap,
\[
\cK_G(g) - \cK_G^* \leq \eta = \frac{\cK_G(g)}{-\min\limits_l \ \partial \cK_G(g)/\partial g_l} (-\min_l \frac{\partial \cK_G(g)}{\partial g_l} - \cK_G(g)) \\
= \cK_G(g) + \frac{(\cK_G(g))^2}{\min\limits_l \ \partial \cK_G(g)/\partial g_l} \\
\]
Remember that the derivative of $\cK_G(g)$ is negative, and with simple arrangement we get,
\[
\frac{\cK_G(g)}{\cK_G^*} \leq \frac{-\min\limits_l \frac{\partial \cK_G(g)}{\partial g_l}}{\cK_G(g)} = \frac{-\min\limits_l ( -n||L^+ b_l||_2^2 )}{n \Tr L^+} = \frac{\max\limits_l ||L^+ b_l||_2^2}{\Tr L^+}
\]

Inserting $g_{\ell w}$, we denote this bound, the \emph{duality gap} bound, by $\alpha_{dual}$,
\[
\alpha_{A,D} \leq  \alpha_{dual}(\ell w) \coloneqq \frac{\max\limits_l ||L_{\ell w}^+ b_l||_2^2}{\Tr L_{\ell w}^+}
\]

\subsubsection{Approximation Ratio UB for the ERMP}

We can now define our derived UB for the approximation ratio - $\alpha_{A,D}$:
\[
    \alpha_{A,D} \leq \alpha_{min} \coloneqq \min(\alpha_1 , \alpha_2),
\] 

Throughout the paper, $\alpha_{AM}, \alpha_1$ always refers to $\alpha_{AM}(\ell w)$, and $\alpha_{dual}, \alpha_2$ always refers to $\alpha_{dual}(\ell w)$, while simply $\alpha$ will refer to $\alpha_{min}$ unless stated otherwise.\\

We can now propose the following, stronger, conjecture:
\begin{conjecture}\label{conj_a_min}
    For any graph $G$, \( \alpha_{min} \leq O(\log n) \)
\end{conjecture}
We emphasize that this is not equivalent to \Cref{conj_LW_apx_ERMP}, as we don't know whether $\alpha_{min}$ is tight. However, our simulations indicate that $\alpha_{min}$ is sufficient (see ahead).\\

For the rest of this section, we derive some formulas on $\alpha_1$ and $\alpha_2$ and their relation. Our first lemma is the connection of the two bounds:
\begin{lemma}\label{lem_a2_der_log_a1}
    $\alpha_2(g) \ = \left\lVert - \nabla_g (\log \alpha_1) \right\rVert_\infty $
\end{lemma}

\begin{proof}
    Recall that \(\alpha_1(g) = \frac{2}{(n-1)^2} \Tr L_g^+ \).
    Using the chain rule and equation~\eqref{equ_der_L_g_plus}, we get that
    \begin{align*}
        - \frac{\partial }{\partial g_l} \log(\alpha_1(g)) &= - \frac{1}{\alpha_1(g)} \frac{\partial \alpha_1}{\partial g_l} = \frac{-(n-1)^2}{2\Tr L_g^+}\frac{2}{(n-1)^2} \Tr \frac{\partial L_g^+}{\partial g_l} \\
        &= \frac{-1}{\Tr L_g^+} \Tr( -L_g^+ b_l b_l^T L_g^+ ) = \frac{\Tr b_l^T L_g^+ L_g^+ b_l}{\Tr L_g^+} = \frac{||L_g^+ b_l||^2}{\Tr L_g^+}
    \end{align*}
    Hence,
    \[
        \left\lVert - \nabla_g (\log \alpha_1) \right\rVert_\infty = \max_l \frac{||L_g^+ b_l||^2}{\Tr L_g^+} = \alpha_2(g) .
    \]
\end{proof}

We can also give a more `combinatorical` interpretation for $\alpha_1$:
\begin{lemma}\label{lem_apx_diam_UB}
    \( \alpha_1 \leq D \), where $D$ is the diameter of the graph.
\end{lemma}
\begin{proof}
    We know that at $g_{\ell w}$, the effective resistance on any edge is $n-1$. 
    Let SP($i,j$) be the shortest path between $i,j$. We can bound $R_{ij}^{\ell w} = R_{ij}(g_{\ell w})$ using the triangle inequality (see equation~\eqref{equ_ER_tri_ineq}):
    \[
        R_{ij}^{\ell w} \leq \sum_{l \in \text{SP}(i,j)} R_l^{\ell w} = \sum_{l \in \text{SP}(i,j)} (n-1) = (n-1)|\text{SP}(i,j)| \leq (n-1)D ,
    \]
    where $D$ is the diameter of the graph. Thus we get that,
    \[
        \cK_G^{\ell w} = \sum_{i<j} R_{ij}^{\ell w} \leq \sum_{i<j} (n-1)D = {n \choose 2} (n-1)D = \frac{n(n-1)^2}{2} D ,
    \]
    and,
    \[
        \alpha_1 = \frac{2}{n(n-1)^2} \cK_G^{\ell w} \leq \frac{2}{n(n-1)^2} \cdot \frac{n(n-1)^2}{2} D = D
    \]
\end{proof}

In contrast, $\alpha_{dual}$ has a close connection to the \emph{condition number} of the Laplacian:
\begin{lemma}\label{lem_a2_cond_num_UB}
    \( \alpha_{dual} \leq \kappa(L_{lw}) = \frac{\lambda_{n}(L_{lw})}{\lambda_{2}(L_{lw})}\)
\end{lemma}
\begin{proof}
    By the Courant-Fisher principle,
    \[
        \lambda_{n}(L_{lw}^+) = \underset{x \in \R^n}{\max} \left\{ \frac{x^T L_{lw}^+ x}{x^T x} \right\}
    \]

    In particular, taking $x = (L_{lw}^+)^{1/2} b_l$ (remember that $L_g^+$ is symmetric PSD matrix) gives us,
    \[
        \frac{b_l^T L_{lw}^+ L_{lw}^+ b_l}{b_l^T L_{lw}^+ b_l} \leq \lambda_{n}(L_{lw}^+)
    \]
    i.e,
    \[
        ||L_{lw}^+ b_l ||^2 = b_l^T L_{lw}^+ L_{lw}^+ b_l \leq  \lambda_{n}(L_{lw}^+) \cdot b_l^T L_{lw}^+ b_l = (n-1)\lambda_{n}(L_{lw}^+).
    \]
    
    In addition, remember that, \( \min(x_1,...,x_k) \leq AM(x_1,...x_k)\). So, using it for the eigenvalues of $L_{lw}^+$, and the fact that the trace is the sum of the eigenvalues, we get that,
    \[
        \lambda_{2}(L_{lw}^+) \leq \frac{1}{n-1} \Tr L_{lw}^+
    \]
    Thus,
    \[
        \alpha_2 = \max_l \frac{||L_{lw}^+ b_l||^2}{\Tr L_{lw}^+} \leq \max_l \frac{(n-1)\lambda_{n}(L_{lw}^+)}{\Tr L_{lw}^+} \leq \frac{(n-1)\lambda_{n}(L_{lw}^+)}{(n-1)\lambda_{2}(L_{lw}^+)} = \frac{\lambda_{n}(L_{lw}^+)}{\lambda_{2}(L_{lw}^+)} = \frac{\lambda_{n}(L_{lw})}{\lambda_{2}(L_{lw})} = \kappa(L_{lw})
    \]
    when we used that fact the eigenvalues of $L^+$ are the reciprocal of the eigenvalues of $L$.
\end{proof}


$ $\\
As promised, we use $\alpha_{dual}$ to derive the approximation ratio for trees: 
\begin{lemma}\label{lem_a2_bowtie}
    For the Bowtie graph, $\mathcal{B}_{n,n,n}$, $\alpha_{dual} \approx 3.12$
\end{lemma}

\begin{proof}
    Recall that for trees
    \[
        \cK_T(g) = \sum_l \frac{c_l}{g_l} ,
    \]
    and,
    \[
    \frac{\partial \cK(g)}{\partial g_l} = -\frac{c_l}{g_l^2}
    \]
    So, we can write $\alpha_{dual}$, for trees, as
    \begin{flalign*}
    &&\alpha_{dual} = \underset{l}{\max} \left( \left. -\frac{\partial \cK(g)}{g_l}/\cK(g)  \right|_{g_{\ell w}} \right) = \underset{l}{\max} \frac{m^2 c_l}{m \sum_k c_k} = \frac{m \cdot c_{\max}}{\sum_k c_k} &&\text{$g_{\ell w} = g_{uni}$, see claim~\ref{clm_LW_tree}}
    \end{flalign*}
    
    Clearly, $c_{max} \leq (3n/2)^2$ (the middle edge). As for the denominator, an easy calculation gives us,
    \begin{align*}
        \sum_k c_k &= \underbrace{2n\cdot m}_\text{$2n$ leafs edges} + \underbrace{\sum_{k=n+1}^{2n-1} k(3n-k)}_\text{path's edges}\\
        &= 2n\cdot m + \sum_{k=1}^{2n-1} k(3n-k) - \sum_{k=1}^{n} k(3n-k)  \approx 2n\cdot m + \frac{10}{3}n^3 - \frac{7}{6}n^3 = \frac{13}{6}n^3 + 6n^2.
    \end{align*}
        
    Thus, we get that
    \[
    \alpha_{dual} (\mathcal{B}_{n,n,n}) \leq \underbrace{\frac{m(3n/2)^2}{(13/6)n^3} \leq \frac{27/4}{13/6}}_{m = 3n-1} \approx 3.12
    \]
\end{proof}

We should note that this is an UB, and experimental results shows a slightly better ratio (see ahead).\\

While we haven't managed to prove our conjecture for general graphs, we can use the above formulas to show it for certain families. For instance, for low-diameter graphs, lemma~\ref{lem_apx_diam_UB} proves our conjecture. In addition, we know that good expanders graphs have a low condition number \cite{spielman11}, thus lemma~\ref{lem_a2_cond_num_UB} proves the conjecture for them as well. For more graphs, we provide a strong evidence for our conjecture, using various simulations, more on that below.

\newpage 

\section{Experimental Results}\label{sec_experiments}

\paragraph{Setup.} The simulation computes the approximation ratio using $\alpha_{min}$, thus matching \Cref{conj_a_min}. All random graphs results are taken to be the maximum of $100$ runs. We compute the LW of the graph using \Cref{alg_LW_computation_ccly} for $\epsilon = 0.01$. Whenever the graph isn't connected or simple, we remove all self loops and take its largest connected component (this is just a technicality).
We simulate our conjecture for several elementary graph families, e.g $k$-regular graphs, lollipop, and grids. 
\Cref{tab_dataset_elem} summarizes the (interesting) graphs we checked, and the approximation ratio of these graphs.
\Cref{fig_simulation_charts} shows an asymptotic behavior of regular graphs and lollipop.

\begin{savenotes}
\begin{table*}[ht!]\small
\normalfont
\footnotesize
  \setlength{\tabcolsep}{6pt}
  \centering
    \vspace{-1em}
    \begin{tabular}{rrn{9}{0}n{9}{0}n{9}{0}r}
    \toprule
      \multicolumn{4}{c}{{\bfseries Dataset}}& \multicolumn{1}{c}{}\\ 
    {\bfseries type} & {\bfseries graph} & \multicolumn{1}{r}{\bfseries nodes} & \multicolumn{1}{r}{\bfseries edges} & \multicolumn{1}{c}{\bfseries $\alpha_{min}(G)$} \\
    \midrule
    \bfseries random, high diameter, $d$-regular & $3$-regular  & 400 & 600 & $1 \period 55$ \\
     & $4$-regular  & 400 & 800 & $1\period 17$ \\
     & $5$-regular & 400 & 1000 & $1\period 11$ \\ 
     & $6$-regular  & 400 & 1200 & $1\period 08$ \\
    \midrule
    \bfseries small-world graph & Watts–Strogatz small-world graph\footnote{with parameters -- $k=4$, $p=2/3$} & 400 & 800 & $1\period 64$ \\
    
    \midrule
    \bfseries grid & balanced grid\footnote{$2$-dimensional square} & 400 & 760 & $1\period 35$ \\
     & 10x-grid\footnote{rectangle with width size $10$} & 400 & 750 & $1 \period 94$ \\
    \midrule
    \bfseries expanders & Margulis-Gabber-Galil graph & 400 & 1520 & $1\period 06$ \\
     & chordal-cycle graph & 400 & 1196 & $1\period 64$ \\
    \midrule
    \bfseries dense graphs & (400,400)--lollipop & 800 & 80200 & $3\period 03$ \\
    
    \midrule
    \bfseries trees & Bowtie & 3000 & 2999 & $2\period 5$ \\ 

    \midrule
    \bfseries real-world graphs\footnote{graphs taken from \cite{MSJ12}} & \textsc{Yeast} & 2224 & 6609 & $2\period 23$ \\
    & \textsc{Stif} & 17720 & 31799 & $4\period 08$ \\
    & \textsc{royal} & 2939 & 4755 & $8\period 96$ \\
    \bottomrule
  \end{tabular}
  \caption{Summary of approximation bounds for elementary graphs.
  \label{tab_dataset_elem}}
\end{table*}
\end{savenotes}

\begin{figure}[H]
  \centering
  \begin{tabular}[c]{lr}
    \begin{subfigure}[c]{0.4\linewidth}
      \includegraphics[width=\textwidth]{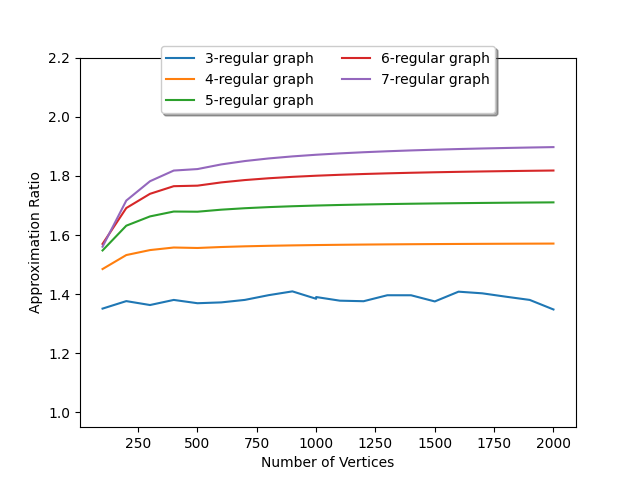}
      \caption{$\alpha_{min}$ of high diameter regular graphs}
      \label{fig_simulation_charts_regular}
    \end{subfigure}&
    \begin{subfigure}[c]{0.4\linewidth}
      \includegraphics[width=\textwidth]{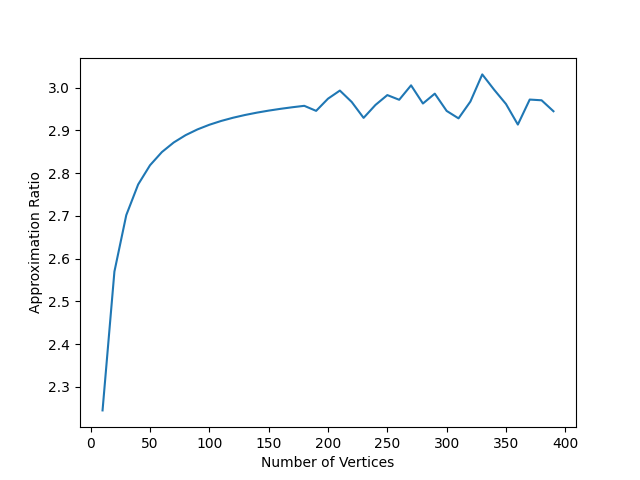}
      \caption{$\alpha_{min}$ of lollipop graph up to $800$ vertices}
      \label{fig_simulation_charts_lolipop}
    \end{subfigure}\\
  \end{tabular}    
  \caption{Asymptotic behavior of elementary graphs}
  \label{fig_simulation_charts}
\end{figure}

\newpage

\section{Spectral Implications of $\linf$-LW Reweighting (Proof of  \Cref{thm_spectral_LW})}\label{sec_LW_spectral}

In this section we study how $\linf$-LW affect the eigenvalue distribution of the reweighted graph Laplacian. We present  several results in this direction, which are of both mathematical and algorithmic interest.

\subsection{
Bound The Algebraic Connectivity}

A classic result of \cite{Mohar91} asserts that for unweighted graphs the algebraic connectivity is bounded by 
\[
\lambda_2 \geq 4/nD.
\]
We generalize this bound for weighted graphs in the following way:
\begin{lemma}
    Given a graph $G(g)$, denote the maximal edge-ER in $R_{max}(g) = \max\limits_l R_l(g)$, and the maximal pairwise ER in $\Tilde{R_{max}}$. Then,
    \begin{flalign}
        &&\lambda_2(L_g) \geq \frac{2}{n \Tilde{R_{max}}} \geq \frac{2}{nD \cdot R_{max}(g)} &&\\
        \text{and}&& && \nonumber \\ 
        && \lambda_2(L_g) \geq \frac{4}{n \sum_l R_l(g)}        
    \end{flalign}
    
\end{lemma}
\begin{proof}
    We know that ER is a metric, so for any two vertices $i,j$,
    \[
    R_{ij}(g) \leq \sum_{l \in P_{ij}} R_l(g) \leq D\cdot R_{max}(g)
    \]
    where $P_{ij}$ is the shortest (length) path between $i$ and $j$. In particular, $\Tilde{R_{max}}(g) \leq D \cdot R_{max}$. Thus
    \[
    \cK_G(g) = \sum_{i<j} R_{ij}(g) \leq {n \choose 2} \cdot \Tilde{R_{max}} = \frac{n(n-1)}{2} \Tilde{R_{max}}(g)
    \]
    Using the fact that $\cK_G(g) = n \Tr L_g^+$, we get
    \[
    \Tr L_g^+ \leq \frac{(n-1)}{2} \Tilde{R_{max}}(g)
    \]
    We know that $\lambda_{max}(L_g^+) \leq \Tr L_g^+ $, so
    \[
    \lambda_{max}(L_g^+) \leq \frac{(n-1)}{2}\Tilde{R_{max}}(g) \implies \lambda_2 \geq \frac{2}{(n-1)\Tilde{R_{max}}(g)} \geq \frac{2}{n\Tilde{R_{max}}(g)} \geq \frac{2}{n D \cdot R_{max}(g)}
    \]

    For the second bound, define the characteristic function of $P_{ij}$ as
    \[
    \chi_{ij}(e) = \twopartdef{1}{e \in P_{ij}}{0}{otherwise}
    \]
    Now, we can write
    \begin{align*}
        \cK(g) &= \frac{1}{2}\sum_{i,j} R_{ij}(g) \\
        &\leq \frac{1}{2} \sum_i \sum_j \sum_{l \in P_{ij}} R_l = \frac{1}{2} \sum_i \sum_j \sum_{l \in E} R_l \cdot \chi_{ij}(e) \\
        &=\frac{1}{2}  \sum_{l} R_l \sum_i \sum_j \chi_{ij}(e) \leq \frac{n^2}{4} \sum_{l} R_l
    \end{align*}
    where in the last line we use the fact that any edge $e$ belongs to at most $\frac{n^2}{4}$ of the paths $P_{ij}$ \cite{Mohar91}. Following the same arguments from before, we get that
    \[
        \lambda_{max}(L_g^+) \leq \frac{n}{4} \sum_{l} R_l \implies \lambda_2 \geq \frac{4}{n \sum_l R_l}
    \]
        
\end{proof}

For sanity check, we should check what we get for unweighted graphs. Indeed for $g=\buni$, we have
\begin{align*}
    &L_g = BB^T \implies L_g^+ = (BB^T)^+ \\ 
    &\implies B L_g^+ B^T = B(BB^T)^+ B^T = BB^+ \\
    &\implies R_l = \diag( B L_g^+ B^T)_l = \diag(BB^+)_l \leq \buni
\end{align*}

where the last line justify because $BB^+$ is projection matrix. So, we got that $R_{max} \leq 1$, thus (using the first bound) 
\[
\lambda_2 \geq \frac{2}{nD \cdot R_{max}} \geq \frac{2}{nD}
\]
which is almost the previous bound (and is better than the weaker version of $1/nD$). \\

\begin{remark} 
    We discussed here only on the algebraic connectivity, but there are many other graph properties and bounds that applies only to unweighted graphs. The above technique, and possibly more, may be used to generalize those bounds to weighted graphs, in terms of ER. 
\end{remark}

\subsection{Minimize the Mixing Time of a Graph}

Recall that the spectral gap of a graph is the smallest non-negative eigenvalue of the Laplacian. Since our graphs are assumed to be connected, this equals $\lambda_2(L_g) = \lambda_2$. Now, a smaller (faster) mixing time is equivalent to higher spectral gap, thus our goal is to maximize $\lambda_2$. 
It is well known that the maximum of $n$ variables can be "smoothed" by using the \emph{LogSumExp} (LSE) function \cite{Nesterov03} :
\[
\LSE(x_1,...,x_n) = \log \left(\sum_i e^{x_i} \right)
\]
The LSE is differentiable, convex function and for any $n$ numbers the LSE satisfy
\[
x_{max} \leq \LSE(x_1,...,x_n) \leq x_{max} + \log(n)
\]
Thus, in an attempt to analyze $\lambda_2 = 1/ \lambda_n(L^+)$, it is natural to consider the LSE of $L^+$ eigenvalues (maximize $\lambda_2$ = minimize $\lambda_n(L^+)$ ). Using the fact that
\[
\Tr \exp(M) = \sum_i e^{\lambda_{i}(M)}
\]
we define the softmaxEV (SEV) function as:
\[
\SEV(g) = \log(\Tr[\exp(L_g^+)])
\]
In order to simplify calculations, exploiting the fact that \emph{log} is a monotone function, we will analyze:
\[
f(g) = \Tr[ \exp(L_g^+)]
\]
Note that $f(g)$ is still convex as exponent preserve convexity. In Appendix \cref{appendix_LSE_opt_crit} we prove that the optimality criterion for SEV is:
\begin{align}\label{equ_sev_opt_crit}
    \Tr [ e^{L_g^+}L_g^+ (I - b_l b_l^T L_g^+)] \geq 0.
\end{align}
In the next section, we show that when \eqref{equ_sev_opt_crit} is \emph{pointwise nonnegative}  (i.e., all summands are PSD ), the criterion \eqref{equ_sev_opt_crit} is tightly connected to the $\linf$-LW optimality condition, hence provides a condition under which LW reweighting improves the mixing-time of the graph $G$.    

\paragraph{A Sufficient condition for optimality}
Note that equation~\eqref{equ_sev_opt_crit} is easily satisfied if \(\bI - b_l b_l^T L_g^+\) is a PSD matrix. Moreover the spectrum of this matrix is quite elegant - \(b_l b_l^T L_g^+\) is one-rank matrix with one non-zero eigenvalue equals to $R_l(g)$ (corresponds to the eigenvector $b_l$). So the spectrum of \(\bI - b_l b_l^T L_g^+ \) is simply
\[
\underbrace{1,...,1}_{n-1} , 1 - R_l(g)
\]
Thus, a sufficient condition for $g$ to be optimal is that the effective resistance on each edge is less than $1$. Note that this condition resembles the normalized LW condition, where the ER on each edge is (at most) $n-1$.\\
Unfortunately, there isn't feasible solution that holds this condition. To see why, note that for any feasible $g$ (equation~\eqref{equ_logdet_grad_g_prod})
\[
 -\nabla \logdet(g)^T \cdot g = \sum_l R_l(g) \cdot g_l = n-1
\]
But, if for some $g'$ we have $R_l(g') \leq 1$ then,
\[
n-1 = \sum_l R_l(g)g_l \leq \sum_l g_l = 1
\]
Contradiction\footnote{$n \leq 2$ is negligible case}. In conclusion, there is no feasible solution such that \(I - b_l b_l^T L_g^+\) is PSD. However this is not means that no optimal solution exist - the trace can be positive even if the matrix is not PSD. Unfortunately, we don't know how to derive closed-form expression for the optimal solution besides the above condition. 

\paragraph{Relative optimality}
While it isn't clear what is the optimal solution for the SEV function, we can use the convex property to try to prove that "solution 1 is better than 2". Formally, for a convex function $f$ 
\[
f(x) \leq f(y) \iff \nabla f(x)^T(y-x) \geq 0
\]

For example, taking $x=g_{lw}=\hat{g} \ , \ y=g_{uni}=(1/m)\buni$, we get the following condition
\[
\nabla f(\hat{g})^T((1/m)\boldsymbol{1} - \hat{g}) = (1/m)\sum_l \frac{\partial f}{\partial g_l} - \nabla f(\hat{g})^T \hat{g} \geq 0
\]

Using the derived formulas from \cref{appendix_LSE_opt_crit}, this gets the form
\[
-(1/m)\sum_l \Tr [e^{L_{lw}^+} L_{lw}^+ b_l b_l^T L_{lw}^+] +\Tr [ e^{L_{lw}^+} L_{lw}^+] = \Tr [e^{L_{lw}^+} L_{lw}^+ (\bI - \sum_l (1/m)b_l b_l^T L_{lw}^+)] = \Tr [e^{L_{lw}^+} L_{lw}^+ (\bI - L_{uni} L_{lw}^+)]  \geq 0
\]

\begin{remark}
The last condition is not a trivial one. To see why, note that if $\bI - L_{uni} L_{lw}^+$ is PSD then the last condition satisfies. It's easy to see that
\[
\bI - L_{uni} L_{lw}^+ \succeq 0 \iff  L_{uni} L_{lw}^+ \preceq \bI \iff  (1/m) B^T L_{lw}^+ B \preceq \bI \iff B^T L_{lw}^+ B \preceq m\bI .
\]
multiplying by $W = 1/(n-1) \cdot \diag(w_\infty)$ on both sides, and we get that
\[
 W^{1/2} B^T (B W B^T)^+ B W^{1/2} = \Pi_{W^{1/2} B^T} \preceq m W
\]
It is known that the eigenvalues of a projection matrix are either $0$ or $1$. On the other hand, the $i$'th eigenvalue of $m W$ is clearly $m (w_\infty)_i /(n-1)$. In other words, if $w_\infty \geq (n-1)/m$ then we can conclude that the mixing time of the reweighted graph is faster than the unweighted graph. Unfortunately, this can never happen -- remember that the sum of LW is $\rank(B)=n-1$, so if $(w_\infty)_{\min} > (n-1)/m$ we get
\[
n-1 \sum_i (w_\infty)_i \geq m \cdot (w_\infty)_{\min} > m \cdot ((n-1)/m) = n-1
\]
contradiction, unless $g_{\ell w} = g_{uni}$. We emphasize that this doesn't mean that LW don't improve the mixing time, since $\bI - L_{uni} L_{lw}^+ \succeq 0$ is just a sufficient condition, and simulations indeed show that LW improves the mixing time. However, we don't have any closed form condition for this improvement. 
\end{remark}


\subsection{Optimal Spectrally-Thin Trees}

Our final application of LW is in spectral thin trees. A spanning tree $T$ of $G$ is $\gamma$-spectral-thin if $L_T \preceq \gamma L_G$. In \cite{AGM10} it has been shown that there is a deep connection between Asymmetric TSP and spectrally-thin trees as finding a low STT can be used with the fractional solution of ASTP to find an approximate solution. This has been the key ingredient in \cite{AGM10} and followup work. In \cite{HO14} the authors showed that for a connected graph $G = \langle V,E \rangle$, any spanning tree $T$ of $G$ has spectral thinness of at least $O(\max\limits_{e \in T} R_G(e))$ (i.e. the maximal edge ER of $T$ in $G$). Furthermore, it is possible to find the tree that achieve this LB in polynomial time. Now, we know (\cref{sec_new_char_LS_ER}) that LW minimizes the latter. So, it is natural to ask how LW can be used to find the optimal spectral thin tree of a graph. 

Following the work of \cite{HO14}, we prove the following lemma:
\begin{lemma}\label{lem_spectral_thin_tree}
    For any connected graph $G = \langle V,E \rangle$ there is a weighted spanning tree $T_g$ such that $T_g$ has spectral thinness of $((n-1)/m) \cdot  O(\log n / \log \log n)$.
\end{lemma}

The proof is very similar to \cite{HO14} and will make use of the following lemma: 
\begin{lemma}\label{lem_independent_set}
    Let $w_1,\dots,w_m \in \R^n$ be an $n$-dimensional vectors with unit norm. Let $p_1,\dots,p_m$ be a probability distribution on these vectors such that the covariance matrix is $\sum_i p_i w_i w_i^T = (1/n) \bI$. Then, there is a polynomial algorithm that computes a subset $S \subseteq [m]$ such that $\{ w_i \mid i \in S \}$ forms a basis of $\R^n$, and $|| \sum_{i \in S} w_i w_i^T || \leq \alpha$ where $\alpha = O(\log n / \log \log n)$.
\end{lemma}

\begin{proof}[Proof of Claim~\ref{lem_spectral_thin_tree}]
\normalfont
Let $w_\infty$ be the LW of the graph. Define $g_0 = \frac{n-1}{m}\buni$ , $\overline{w} = \frac{n-1}{m}w_\infty$. For each $e_l \in E$, define $w_l = \sqrt{(g_0)_l} \cdot (L_{\overline{w}})^{+/2} b_l$, $\boldsymbol{p} = \frac{1}{n-1}w_\infty$. Indeed, by ~\eqref{equ_LS_facts} and the multiplicative property of pseudo-inverse we have
\begin{align*}
    \sum_l p_l = 1 \ , \ ||w_l||^2 = (g_0)_l \cdot \frac{m}{n-1} b_l^T (B W_\infty B^T)^+ b_l = \frac{n-1}{m} \cdot \frac{m}{n-1} = 1
\end{align*}
In addition,
\begin{align*}
    \sum_{e_l \in E} p_l w_l w_l^T &= \frac{1}{n-1} (L_{\overline{w}})^{+/2} \left (\sum_l (w_\infty)_l (g_0)_l b_l b_l^T \right) (L_{\overline{w}})^{+/2} \\
    &= \frac{1}{n-1} (L_{\overline{w}})^{+/2} \left( \sum_l \overline{w}_l b_l b_l^T \right) (L_{\overline{w}})^{+/2} \\
    &= \frac{1}{n-1} \bI_{\text{im} L_G}
\end{align*}
So, we can view the vectors $\{w_l \mid l \in E \}$ as $(n-1)$-dimensional vectors (in the linear span of $B$) and apply \cref{lem_independent_set} to get a set $T \subset E$ of size $n-1$ ($T$ forms a basis in $\R^{n-1}$) such that $\{ w_e \mid e \in T\}$ is linearly independent and
\[
\sum_{e_l \in T} w_l w_l^T \preceq O(\log n/\log \log n) \bI_{\text{im} L_G} 
\]
The first two conditions imply that $T$ induced a spanning tree of $G$. Then, the last condition gives us (simple rearrangement)
\[
\sum_{e_l \in T} (g_0)_l b_l b_l^T = L_{T_g} \preceq O(\log n/\log \log n) L_{\overline{w}}
\]
Now, we know LW are at most $1$, so, $\overline{w} = \frac{n-1}{m}w_\infty \leq \frac{n-1}{m} \buni$, implying that
\[
L_{\overline{w}} \preceq \frac{n-1}{m} BB^T = \frac{n-1}{m} L_G
\]
Thus we got that
\[
L_{T_g} \preceq ((n-1)/m) \cdot O(\log n / \log \log n) L_G
\]
\end{proof}

\paragraph{The ATSP Angle}
It would be interesting to understand if and how can  \emph{weighted STTs} be used for ATSP  rounding schemes, a-la \cite{AGM10}. Nevertheless, we emphasize two advantages of LW-weighted STTs:  First, in contrast to the unweighted case  \cite{AO15},
we guarantee to find an $O((n-1)/m)$-STT 
\emph{regardless} of the original graph, which is optimal -- to see why, recall that LW is the optimal minimizer of the maximal edge-ER. Taking $g_{\ell w}$ and $g_{uni}$, both normalized, we have 
\[
R_{max}(g_{\ell w}) = n-1 \leq R_{max}(g_{uni}) = m \cdot R_{max}(\buni) \implies R_{max}(\buni) \geq \frac{n-1}{m}.
\]
Second, the unweighted STT guaranteed by \cite{AO15} can at best achieve this bound with total weight $n-1$. By contrast, the total weight of our weighted STT is $\frac{(n-1)^2}{m} \leq n-1$ (and for most graphs $\ll$). In this sense, the tree from Lemma \ref{lem_spectral_thin_tree} is always "cheaper".

\section{Computing Lewis Weights}\label{sec_computing_LW}

In this section we outline recent accelerated methods for computing $\linf$-LW of a matrix via \emph{repeated leverage-score} computations (in out case, Laplacian linear systems). 
The first method, due to \cite{ccly} (building on \cite{CP14}), runs in input-sparsity time and is very practical, but provides only low-accuracy solution. 
The second method, due to \cite{flps21}, 
provides a \emph{high-accuracy} algorithm 
for $\ell_p$-LW, using $\tilde{O}(p^3 \log (1/\eps))$ leverage-score computations (Laplacian linear systems in our case).   
Fortunately, as we show below, low-accuracy is sufficient the ERMP application, hence we focus on the algorithm of \cite{ccly}. 
However, we remark that in most optimization applications, 
 computing $\ell_p$-LW for $p=n^{o(1)}$ is suffices (see \cite{flps21,LS14}) in which the \cite{flps21} algorithm yields a high-accuracy $O(m^{1+o(1)}\log(1/\eps))$ time algorithm for Laplacians.


\subsection{Computing $\linf$-Lewis Weights to Low-Precision \cite{ccly}}


\begin{algorithm}
\caption{Computing $\linf$-Lewis weight}
\label{alg_LW_computation_ccly}

\hspace*{\algorithmicindent}\textbf{Input:} A matrix $A \in \R^{m \times n}$ with rank $k$, $T$ - number of iterations. \\
\hspace*{\algorithmicindent}\textbf{Result:} $\linf$-LW, $w \in \R^m$. \\
\begin{algorithmic}
\State Initialize: $w_l^{(1)} = \frac{k}{m}$ for $l=1, \dots ,m$. 
\For{$t = 1, \dots ,T-1$} 
\State $w_l^{(t+1)} = w_l^{(t)} \cdot a_l^T (A^T \diag(w^{(t)})A)^+ a_l$ ; for $l=1, \dots, m$ 
\Comment{// We can use here a Laplacian LS solver.}
\EndFor
\State $(w)_i = \frac{1}{T} \sum\limits_{t=1}^T w_i^{(t)}$ for $i=1,\dots ,m$ \\
\Return $w$
\end{algorithmic}
\end{algorithm}

Since we use \cite{ccly} in a black-box, we only give a high level overview of Algorithm \ref{alg_LW_computation_ccly}. 
The basic idea of \cite{ccly, CP14} is 
to use the observation that for (exact) $\linf$-LW, we have 
\begin{align}\label{equ_ccly_alg_fixed_equ}
    w_i = \tau_i(W^{1/2} A). 
\end{align}
\cite{ccly} use \eqref{equ_ccly_alg_fixed_equ} to define a fixed point iteration described in \Cref{alg_LW_computation_ccly}. In other words, the algorithm repeatedly computes the leverage scores of the weighted matrix, updating the weights according to the average of past iterations, until an (approximate) fixed-point is reached. The performance of this algorithm is as follows:
Let $\epsilon >0$ and denote by $\tilde{w}$ the output of the algorithm for input $A \in \R^{m \times n}$. Our goal is to compute the approximate LW, $\tilde{w}$, such that 
\begin{align*}
 \tilde{w} \approx_\epsilon w_\infty  \numberthis \label{equ_mul_apx_LW}
\end{align*}
where $a \approx_\epsilon b$ iff $a = (1 \pm \epsilon)b$. However, the approximation guarantee of \cref{alg_LW_computation_ccly} is in the \emph{optimality} sense, meaning  
\begin{align*}
\tau_i(\tilde{W}^{1/2} A) \approx_\epsilon \tilde{w}_i  , \ \  \forall i \in [m]. \numberthis \label{equ_opt_apx_LW}
\end{align*}
Fortunately, \cite{flps21} recently showed that $\epsilon$-approximate optimal LS \eqref{equ_opt_apx_LW} imply 
$\epsilon$-approximate LW \eqref{equ_mul_apx_LW}: 
\begin{theorem}
    (LS-apx $\Rightarrow$ LW-apx,  \cite{flps21}). For all $\epsilon \in (0, 1)$, Algorithm~\ref{alg_LW_computation_ccly} outputs a ($1+\epsilon$)-approximation of LW with $T = O(\epsilon^{-1} \log \frac{m}{n})$ iterations.
\end{theorem}

\paragraph{Approximate Lewis Weights:}
Since the output of \Cref{alg_LW_computation_ccly} satisfies
$\tilde{w} \approx_\epsilon w_\infty$, 
this implies 
\[
\Tilde{L_{w}} = \sum_l (\tilde{g_{w}})_l b_l b_l^T \approx_\epsilon \sum_l (g_{\ell w})_l b_l b_l^T = L_{\ell w}
\]
where $\tilde{g_{w}} ,g_{\ell w} $ is the weight vector defined by $\tilde{w} ,w_\infty$, respectively. Thus,
\[
\Tr \Tilde{L_{w}} \approx_\epsilon \Tr L_{\ell w}
\]
In other words, computing the approximate LW guarantees  a-$(1+\epsilon)$ factor for the trace of $L_{\ell w}$. Since our objective function is solely the trace, that means approximate LW will gives us a-$(1+\epsilon)$ factor for our approximation ratio. Since we are only shooting for $O(1)$-approximation for ERMP, the latter is sufficient, so \Cref{alg_LW_computation_ccly} can be used to produce the approximate-ERMP weights claimed in Theorems \ref{thm_LW_apx_ERMP_trees} and \ref{thm_UB_general_graphs}. Note that this also applies to the rest of results, as the proximity notion in~\eqref{equ_mul_apx_LW} also implies spectral approximation.



\section{
An $\Omega(n/\log^2 n)$ Separation of A-vs-D Design 
for General PSD Matrices}
\label{sec_optimal_design}



We finish by justifying our focus on Laplacians. We show that, for general PSD matrices, the problem of A- vs D-optimal design are in fact very different, in the sense that the $\linf$-LW cannot provide better than $\Omega(n/\log^2 n)$ approximation to the A-optimal design problem.


The motivation for this separation is a reformulation of A- and D-optimal design as \emph{Harmonic mean} (HM) vs. \emph{Geometric mean} (GM) minimization problem. We show an unbounded gap between the HM and GM of a general sequence (up to some constraint, more on that later). We then construct a simple experiment matrix (i.e a pointset in space) for which the D-optimal solution (given by $\linf$-LW) is an $\Omega(n/log^2(n))$ multiplicative approximation to the A-optimal solution.

\paragraph{A-optimal design as Harmonic mean optimization:}

We begin with A-optimal design. Given an experiment matrix, $V = (v_1, \dots, v_m) \  \in \R^{n \times m}$, the dual problem of A-optimal design is

\begin{equation*}
\begin{array}{ll@{}ll}
\text{maximize}  & \Tr(W^{1/2})^2&\\
\text{subject to}&  v_i^T W v_i \leq 1 \ , \ i=1,\dots,m
\end{array}
\end{equation*}
with $W \in \boldsymbol{S}^n_+$ (see also~\eqref{equ_ERMP_dual_prob}). Note that the constraint define an ellipsoid $\mathcal{E}(W) = \{ x \in \R^n \mid x^T W x \leq 1 \}$ that enclose on the pointset $\{v_i\}_{i=1}^m$. We can write our objective function (neglecting the square since it does not change the optimal solution, just its value) as (see~\eqref{equ_ellipsoid_semiaxis_def}):
\[
\Tr \ W^{1/2} = \sum_{i=1}^n \lambda_i(W)^{1/2} = \sum_{i=1}^n (\sigma_i)^{-1} = \frac{n}{HM(\boldsymbol{\sigma})} ,
\]
where $\boldsymbol{\sigma}$ is the semiaxis lengths of the ellipsoid $\mathcal{E}(W)$.

In other words, the A-optimal design problem attempts to minimize the \emph{Harmonic mean} of the semiaxis lengths of the enclosing ellipsoid $\mathcal{E}(W)$. Thus, A-optimal design is equivalent to the following problem

\begin{equation*}
\begin{array}{ll@{}ll}
\text{minimize}  & HM(\boldsymbol{\sigma}(W))&\\
\text{subject to}&  v_i^T W v_i \leq 1 \ , \ i=1,\dots,m
\end{array}
\end{equation*}

\paragraph{D-optimal design as Geometric mean minimization}
We can do the same for D-optimal design. The dual problem of D-optimal designs is
\begin{equation*}
\begin{array}{ll@{}ll}
\text{maximize}  & \logdet(W) &\\
\text{subject to}&  v_i^T W v_i \leq 1 \ , \ i=1,\dots,m
\end{array}
\end{equation*}
with $W \in \boldsymbol{S}^n_+$. Same as before, we can write our objective function as
\[
\log |W| = \log(\lambda_1(W) \cdots \lambda_n(W)) = \log((\sigma_1 \cdots \sigma_n)^{-2}),
\]
and we get that,
\[
\logdet \ W = \log |W| = -2n\log((\sigma_1 \cdots \sigma_n)^{1/n}) = -2n\log(GM(\boldsymbol{\sigma})),
\]
where we used the same notations as above.

In other words, maximize the function $\logdet(W)$, is the same as minimizing the \emph{Geometric mean} of the semiaxis lengths of the enclosing ellipsoid $\mathcal{E}(W)$. Thus, D-optimal design is equivalent to  the following problem 

\begin{equation*}
\begin{array}{ll@{}ll}
\text{minimize}  & GM(\boldsymbol{\sigma}(W))&\\
\text{subject to}&  v_i^T W v_i \leq 1 \ , \ i=1,\dots,m
\end{array}
\end{equation*}

\begin{remark}
    The geometric interpretation of D-optimal design as the minimal volume enclosing ellipsoid is also known as Outer John Ellipsoid, and is a well studied problem. For further details see also \cite{KT93,ccly,ZF20,todd16,Puk06}.
\end{remark}

\subsection{HM and GM comparison}\label{sec_HM_GM_comp}

As we just saw,  A-optimality and D-optimality are correlated to the Harmonic and Geometric mean of the same value - semiaxis lengths of some enclosing ellipsoid. We know that for any vector $\x \in \R^n$ , \(HM(\boldsymbol{x}) \leq GM(\boldsymbol{x})\) , with equality achieved only at a uniform vector, and that they are continuous functions of $\x$. So, It is natural to think that they are "monotone" in some sense. For example, one might think that for two vectors $\x,\x' \in \R^n$ 
\[
GM(\x) \leq \alpha \ GM(\x') \iff HM(\x) \leq \alpha  \ HM(\x') ,
\]
for some constant factor $\alpha$. Note that this kind of property makes the two problems equivalent (up to maybe some constant factors). 
We show that this is not the case. 
The following technical claim shows that, from any sequence $\x$ of $n$ positive numbers, it is possible to construct a related sequence $\x'$ such that $GM(\x') = GM(\x)$ but $HM(\x') \ll HM(\x)$. We defer the proof to \cref{appendix_prf_clm_hm-gm}.
\begin{claim}\label{clm_HM_GM_gap}
    Let $n \in \mathbb{N}$ , $\x \in \R^n_+$. For any $0<t<1 \in \R$, there exist $\x' \in \R^n_+$ such that $GM(\x')=GM(\x)$ but $HM(\x') = t \cdot HM(\x)$ 
\end{claim}
Note that this claim is a bit stronger than what we need - we don't even have to "pay" at the GM (although we do pay at the AM) in order to decrease the HM \emph{as much as we want}.

\subsection{A Separation of A-Optimal and D-Optimal Design for General PSD Matrices}

As we just said, the HM and GM aren't approximately close for general sequences (i.e there is an unbounded gap between HM and GM). Since we can interpret A-optimal and D-optimal design as "HM vs GM" problem, this naturally raises the following claim
\begin{claim}\label{clm_A_D_gen_gap}
    There exist $V \in \R^{n \times m}_+$ such that 
    \[
        \Tr (Vg^*_D V^T)^{-1} \gg \Tr (Vg^*_AV^T)^{-1}
    \]
    where $g^*_D,g^*_A$ are the optimal solutions for D-optimal and A-optimal design, respectively.
\end{claim}

While we can use claim~\ref{clm_HM_GM_gap} to construct this "counter example", we offer a much simpler instance with $\Tilde{\Omega}(n)$ gap, which is more than enough for showing a separation.   

\begin{proof}
    Let $V$ be a diagonal matrix with $[1,\dots,n]$ on the diagonal. We already mentioned that the optimal solution for D-optimal design is the LW of the experiment matrix $V$, i.e
    \[
        (g^*_D)_l = \left(w(V)_\infty \right)_l.
    \]
    By definition of LW, $g^*_D$ satisfies equation~\eqref{equ_inf_LW}:
    \[
        v_l^T (V^T W_\infty V)^{-1} v_l = 1 ,
    \]
    Where $v_l$ is the $l$'th row of $V$. Using the fact the $V$ is a diagonal matrix whose $l$'th entry equals to $\ell$, we can write the last equation as
    \[
        \ell (\ell^2 (w_\infty)_l)^{-1} \ell = 1 \implies (w_\infty)_l = 1 ,
    \]
    and with normalization we get  
    \[
    g^*_D = (1/n) \buni.
    \]
    
    The value of A-optimal at $g^*_D$ equals  to,
    \[
        \Tr (V g^*_D V^T)^{-1} = \sum_{i=1}^n (i^2/n)^{-1} = n \sum_{i=1}^n \frac{1}{i^2} \approx \frac{\pi^2}{6} n \approx 1.5n.
    \]
    \\
    Next, let's look at $(g_0)_i = 1/i$ (up to normalization). We know that 
    \[
        \sum_{i=1}^n \frac{1}{i} \approx log(n),
    \]
    so, $(g_0)_i \approx \frac{1}{log(n) \cdot i}$. The value of A-optimal at $g_0$ is 
    \[
        \Tr (V g_0 V^T)^{-1} \approx \sum_{i=1}^n \left(\frac{i^2}{i \cdot log(n)}\right)^{-1} = log(n) \sum_{i=1}^n \frac{1}{i} \approx log^2(n).
    \]
    Thus, we get that
    \[
        \Tr (V g^*_D V^T)^{-1} = \theta(n) \gg \Tr (V g_0 V^T)^{-1} = \theta(log^2(n)) \geq \Tr (V g^*_A V^T)^{-1}
    \]
    where that last inequality is from the optimality of $g^*_A$.
    Thus we can conclude that there is a natural gap between A-optimal and D-optimal design of \emph{at least} $\Omega(n/log^2n).$ 
\end{proof}

\section*{Acknowledgement} 
We would like to thank 
Peilin Zhong for insightful discussions and for suggesting the connection between infinite Lewis weights and mixing time of random walks. 
\newpage
\printbibliography

\newpage

\appendix

\section{Appendix}

\subsection{Background: Convex Optimization Analysis}\label{appendix_cvx_opt_recipe}

All the optimization problems discussed in this paper are convex problems. On our analysis of the optimal solution, we will use the following recipe for problems of the following form:
\begin{equation*}
\begin{array}{ll@{}ll}
\text{minimize}  & f(g) &\\
\text{subject to}& \boldsymbol{1}^Tg=1 \ , \ g \in \R_+^m 
\end{array}
\end{equation*}
when $f(g)$ is convex. It is a well known identity, that a feasible $g$ is optimal for $f$ iff
\[
\nabla f(g)^T \cdot (\hat{g} - g) \geq 0 \ , \ \forall \ \hat{g}: \  \buni^T\hat{g}=1.
\]
This is the same as,
\[
\nabla f(g)^T \cdot (e_l - g) \geq 0 \ , \ \forall \ l=1\dots m
\]
(expand $\hat{g}$ to the elementary basis representation). So, the optimality criteria for $g$ is,
\[
\frac{\partial f}{\partial g_l}(g) -\nabla f(g)^T \cdot g \geq 0 \ , \ \forall \ l=1\dots m .
\]
Thus, we will use the following recipe - given a convex problem of the above form:
\begin{enumerate}
    \item Compute the partial derivatives of $f$, i.e $\frac{\partial f(g)}{\partial g_l}$.
    \item Compute $\langle \nabla f(g) , g \rangle = \nabla f(g)^T \cdot g$.
    \item Derive the optimality criteria:
    \begin{align*}
        \frac{\partial f}{\partial g_l} -\nabla f(g)^T \cdot g \geq 0 \ , \ \forall \ l=1 \dots m \numberthis \label{equ_cvx_opt_criteria}
    \end{align*}
\end{enumerate}

\subsection{Background: Experimental Optimal Design}\label{appendix_optimal_design}

Here we give a high level overview of experimental optimal design. For more information see \cite{boyd_convex_2004,Puk06}.

Assume you want to estimate a vector $x \in \R^n$, using the following experiments: Each round $i = 1,\dots, m$, you can choose one \emph{test} vector from a given $p$ possible choices -- $v_1,\dots,v_p$ -- and you get $y_i = v_{j_i}^T x + w_i$, where $j_i \in [p]$ and, $w_i$ is some Gaussian noise (i.e independent Gaussian RV with zero mean and unit variance). It usually assumed that $V = (v_1, \dots, v_p)$ is a full rank matrix, but the it can be easily generalized. The optimal estimation is given by the least-square solution:
\[
\hat{x} = \underbrace{\left(\sum_{i=1}^m v_{j_i} v_{j_i}^T\right)^{-1}}_{E} \ \cdot\sum_{i=1}^m y_i v_{j_i} .
\]
For any estimation $x$, we define the estimation error $e=\hat{x}-x$, with the associated covariance matrix $E$. The goal of optimal design is to minimize $E$ w.r.t some partial order (e.g Loewner order). 
In general, optimal design can be a hard combinatorial problem for large $m$. To see why, let $m_j$ be the number of rounds we chose $v_j$, and express $E$ using $m_j$:
\[
    E =  \left(\sum_{i=1}^m v_{j_i} v_{j_i}^T\right)^{-1} =  \left(\sum_{j=1}^p m_jv_jv_j^T\right)^{-1}
\]
This shows that optimal design can be formulate with the integer variables $m_j$, with the constraint that they will sum up to $m$. If $m$ is quite large, then we might deal with relative small integers, making it quite difficult problem. However, we may relax it in the following way. Define $g_j = m_j/m$, and express $E$ in terms of $g_j$:
\[
    E(g) = \frac{1}{m} \left(\sum_{j=1}^p g_jv_jv_j^T\right)^{-1}
\]
Now, instead of searching over the integers, we let $g_j$ be real numbers, under the constraint that they will sum up to $1$. Thus, the relaxed optimal design asks to minimize $E$ such that $\buni^T g = 1$, w.r.t some partial order. For example, A- and D-optimal design minimizes w.r.t the trace and determinant, respectively. We refer to $V$ as the \emph{experiment} matrix, and denote by $E_V(g) = V\cdot \diag(g) \cdot V^T$, such that $E_V(g)^{-1}$ is the error covariance matrix.

\subsection{The D-Optimal Design Problem}\label{appendix_D_design}

Here we provide an alternative way to solve the D-optimal design problem via a convex analysis approach. We focus on the case of Laplacians, but it can be generalized to any experiment matrix. We follow the recipe described in \ref{appendix_cvx_opt_recipe}.
We define the D-optimal design problem over Laplacian, as a relaxation of problem~\eqref{equ_D_optimal_def}, in the following manner:
\begin{definition}\label{def_LD_g}
Given an undirected graph $G$, denote the Laplacian of the weighted graph $G(g)$ in $L_g$. The \emph{LD} minimization problem asks how to re-weight the edges in order to minimize the \emph{logdet} of $L_g^+$:
\begin{equation*}
\begin{array}{ll@{}ll}
\text{minimize}  & \logdet( L_g^+ + (1/n)\boldsymbol{1}\boldsymbol{1}^T) &\\
\text{subject to}& \boldsymbol{1}^Tg=1 \ , \ g \geq 0
\end{array}
\end{equation*}
\end{definition}

\begin{remark}\label{remark_LD_G_objective}
    One might wonder why this is our objective function. Note that $L_g^+$ has zero eigenvalue while we interesting in only the positive eigenvalues. However, it is not difficult to prove that:
    \[ \lambda(L_g^+ + 1/n \boldsymbol{1}\boldsymbol{1}^T) = \twopartdef {\lambda(L_g^+)} {\lambda(L_g^+) \neq 0} {1} {\lambda(L_g^+) = 0} \]
    So, our objective function is simply the product of all positive eigenvalues of $L^+$. We further abuse this notation to write $|L_g^+ + (1/n)\buni \buni^T|$ as $|L_g^+|$.
\end{remark}

We will write the objective function as $LD_G(g) = \logdet(L_g^+)$. Whenever it's clear that we refer to the graphs version we will simply write $LD(g)$. Note that for a general experiment matrix $V$ the objective function is equals to $LD_V(g) = \logdet(E_V(g)^+)$. This leads us to the following result:
\begin{lemma}\label{thm_D_optimal_LW}
    Given an undirected graph $G$ with edge-incident matrix $B \in \R^{n \times m}$, the weight-vector that minimizes the $LD_G(g)$ function is $LW(B^T)$, up to normalization. More generally, for an experiment matrix $V \in \R^{n \times m}$, the solution for D-optimality for $V$ is 
    $LW(V^T)$, up to normalization.
\end{lemma}


Let's start by computing the partial derivatives. We will use the following fact: Suppose that the invertible symmetric matrix $X(t)$ is a differentiable function of parameter $t \in \R$. Then we have
\begin{equation}\label{equ_der_inv_mat}
    \frac{\partial X^{-1}}{\partial t} = -X^{-1}\frac{\partial X}{\partial t} X^{-1}
\end{equation} 
We also use \emph{Jacobi's formula} - given differentiable matrix $X(t)$:
\begin{equation}\label{equ_jacobi_formula}
     \frac{\partial |X|}{\partial t} = |X| \Tr (X^{-1} \frac{\partial X}{\partial t})
\end{equation}
Now, if we write the one-rank decomposition of $L_g = \sum_{l=1}^m g_l b_l b_l^T$, it is clear that
\begin{equation}\label{equ_der_L_g}
    \frac{\partial L_g}{\partial g_l} = b_l b_l^T
\end{equation}

Using the above, and equation~\eqref{equ_laplac_inv}, we get that,
\begin{align*}
    \frac{\partial L_g^+}{\partial g_l} &=  \frac{\partial (L_g + 1/n\boldsymbol{11}^T)^{-1}}{\partial g_l} = -(L_g + 1/n\boldsymbol{11}^T)^{-1}\frac{\partial L_g}{\partial g_l}(L_g + 1/n\boldsymbol{11}^T)^{-1} && \text{using~\eqref{equ_der_inv_mat} } \\
    &= -(L_g^+ + 1/n\boldsymbol{11}^T)\frac{\partial L_g}{\partial g_l}(L_g^+ + 1/n\boldsymbol{11}^T) && \text{from~\eqref{equ_laplac_inv}}\\
    &= -(L_g^+ + 1/n\boldsymbol{11}^T)b_l b_l^T(L_g^+ + 1/n\boldsymbol{11}^T) && \text{using~\eqref{equ_der_L_g}} \\
    &= -L_g^+ b_l b_l^T L_g^+ \numberthis \label{equ_der_L_g_plus}
\end{align*}  
when we used the fact that $b_l^T \buni = 0$.

Now, using \emph{Jacobi's formula}~\eqref{equ_jacobi_formula}, and basic chain rule we get:
\begin{align*}
     & \frac{\partial |L_g^+|}{\partial g_l} = |L_g^+| \Tr ((L_g^+)^{-1} \frac{\partial L_g^+}{\partial g_l}) 
\\
     &\implies \frac{\partial \log |L_g^+|}{\partial g_l} = |L_g^+|^{-1} \frac{\partial |L_g^+|}{\partial g_l} = \Tr((L_g^+)^{-1} \frac{\partial L_g^+}{\partial g_l})
\end{align*}
   
Combining the last two equations leads to,
\begin{align*}
\frac{\partial \log |L_g^+|}{\partial g_l} & = -\Tr( (L_g^+)^{-1} L_g^+ b_l b_l^T L_g^+ ) \\
& = - \Tr ( (b_l b_l^T L_g^+) = -\Tr (b_l^T L_g^+ b_l) = -R_l(g) \numberthis \label{equ_logdet_der}
\end{align*}

\begin{remark}
    For the general case, we will get that:
    \begin{align*}
        \frac{\partial LD_V(g)}{\partial g_l} & = -\Tr( E_V(g)^{-1} v_l v_l^T ) \\
        & = - \Tr ( (v_l^T E_V(g))^{-1} v_l ) = v_l^T (V \cdot \diag(g) \cdot V^T)^{-1} v_l \numberthis \label{equ_der_LD_g}
    \end{align*}
    and, in terms of LS:
    \begin{equation}\label{equ_der_LD_g_with_LS}
        \frac{\partial LD_V(g)}{\partial g_l} = - g_l^{-1} \cdot \tau_l(W_g^{1/2} V^T) 
    \end{equation}
\end{remark}

Next, we need to compute $\langle \nabla LD(g) , g \rangle = \nabla LD(g)^T \cdot g$ . We will use the following "trick" - by the multiplicativity of determinant, and properties of pseudo-inverse:
\begin{align*}
    LD(\alpha g) &= \log |L_{\alpha g}| = \log | (B \cdot \diag(\alpha g) \cdot B^T)^+| \\
    &=\log |\alpha^{-1} L_g^+| \\
    &= \log ( \alpha^{-(n-1)}|L_g^+|) && \text{because $\rank(L_g) = n-1$} \\
    &= LD(g) - (n-1)\log(\alpha) && \text{logarithm laws}
\end{align*}
Taking the derivative of both sides w.r.t $\alpha$, we get that:
\[ 
\nabla LD(\alpha g)^T \cdot g = \frac{\partial}{\partial \alpha}(LD(g) - (n-1)\log (\alpha)) = -\frac{n-1}{\alpha} 
\]
Evaluate the expression at $\alpha = 1$, yield
\begin{equation}\label{equ_logdet_grad_g_prod}
    \nabla LD(g)^T \cdot g = -(n-1)
\end{equation} 

\begin{remark}
    It's not hard to see that for general matrix $V$, with rank $d$, $rank(E_V) = rank(V) = d$, and $|\alpha^{-1} E_V^{-1}| = \alpha^{-d} |E_V^{-1}|$. So we will have,
    \[
    LD_V(\alpha g) = LD_V(g) - d\log(\alpha) ,
    \]
    and accordingly,
    \[
    \nabla LD_V(g)^T \cdot g = -d.
    \]
\end{remark}

Now, we can plug all the above in the optimality criteria (equation~\eqref{equ_cvx_opt_criteria}), and get that $g$ is optimal for $LD(g)$ iff,
\[ 
\frac{\partial LD(g)}{g_l} + n-1 = -R_l(g) + n-1 \geq 0 \ , \ l=1,\dots, m 
\]
or,
\begin{align*}
R_l(g) \leq  n-1 \ , \ \forall \ l=1,\dots,m . \numberthis \label{equ_LD_G_opt_crit}
\end{align*}
And for general $V$,
\begin{align*}
g_l^{-1} \cdot \tau_l(W_g^{1/2} V^T)  \leq \rank(V) \ , \ l=1, \dots , m  \numberthis \label{equ_LD_opt_criteria}
\end{align*}

Finally, we can show that $\linf$-LW is indeed the optimal solution for $LD(g)$. We will show it for the general case. Let $C = W_g^{1/2} V^T$. Assume $V$ has rank $d$. Let $g_{\ell w} \coloneqq (1/d) w_\infty(V^T)$. First of all, from equation~\eqref{equ_LS_facts}
\[
\sum_{l=1}^m (g_{\ell w})_l = (1/d) \sum_{l=1}^m w_\infty(V^T)_l = (1/d) \cdot d = 1
\]
so $g_{\ell w}$ is feasible. Next, we need to show that it matches the optimality criteria. Recall from equation~\eqref{equ_inf_LW} that for the non-normalized LW,
\[
v_l^T(V W_\infty V^T)^+ v_l = 1 , \ l=1, \dots , m.
\]
Since we normalize by $(1/d)$ we get that for $g=g_{\ell w}$,
\[
v_l^T(V W_g V^T)^+ v_l = d , \ l=1, \dots , m.
\]
Hence, for $g=g_{\ell w}$ we have,
\[
\tau_l(W_\infty^{1/2} V^T) = g_l^{1/2} \cdot v_l^T (V \cdot W_\infty V^T)^+ v_l \cdot g_l^{1/2} = d \cdot g_l , \ \forall \  l=1, \dots , m .
\]
Substituting the latter in equation~\eqref{equ_LD_opt_criteria} gives us
\begin{align*}
g_l^{-1} \cdot \tau_l(W_g^{1/2} V^T) = d \cdot g_l \cdot g_l^{-1} = d \leq d. \numberthis \label{equ_LD_LW_opt_saturation}
\end{align*}
which matches exactly the optimality criteria, and proving \Cref{thm_D_optimal_LW}. Note that, indeed $g_{\ell w}$ \emph{saturates} the optimality criteria.

\subsection{Proof of Claim~\ref{clm_HM_GM_gap}}\label{appendix_prf_clm_hm-gm}

\begin{proof}
    Let $\x = (x_1,x_2,\dots, x_n)$. Define $\x' = (px1,x2/p, x3,\dots,x_n)$. Indeed,
    \[
        GM(\x') = (\Pi_i \ x'_i)^{1/n} = (\Pi_i \  x_i)^{1/n} = GM(\x)
    \]
    Let $S = \sum_{i=1}^n \frac{1}{x_i}$ and similar for $S'$. Clearly, $HM(\x) = n/S$, so to increase the HM by $t$ we should decrease $S$ by $t$.

    Explicitly, we need that
    \[
        S' = \sum_{i=1}^n \frac{1}{x'_i} = S/t = (1/t)\sum_{i=1}^n \frac{1}{x_i}.
    \]
    From our definition of $\x'$ we get that
    \[
        S' = \frac{1}{px_1} + \frac{p}{x_2} + \sum_{i=3}^n \frac{1}{x_i} = \frac{1}{px_1} + \frac{p}{x_2} + S - \frac{1}{x_1} - \frac{1}{x_2} = S + (1/p - 1)\frac{1}{x_1} + (p-1)\frac{1}{x_2} = S/t,
    \]
    which is a quadratic equation in $p$:
    \begin{align*}
        (1-1/t)S - \left(\frac{1}{x_1} + \frac{1}{x_2} \right) +\frac{1}{x_1}\frac{1}{p} + \frac{1}{x_2}p = 0. \numberthis \label{equ_quad_p_t}
    \end{align*}
    We can change variables to $y_1 = \frac{1}{x_1} , y_2 = \frac{1}{x_2} , t' = (1/t - 1)S$ to get:
    \[
        y_2 \cdot p^2 - (t'+y_1+y_2) \cdot p + y_1 = 0,
    \]
    which has the solution:
    \begin{align*}
        p &=\frac{1}{2y_2}\left[(t' + y_1 + y_2 ) + \sqrt{(t'+y_1+y_2)^2 - 4y_2y_1} \right] \\
        &= \frac{1}{2y_2}\left[(t' + y_1 + y_2) + \sqrt{t'^2 + (y_1+y_2)^2 + 2t'(y_1+y_2) - 4y_2y_1} \right] \\
        &= \frac{1}{2y_2}\left[(t' + y_1 + y_2) + \sqrt{t'^2 + 2t'(y_1+y_2) + (y_1 - y_2)^2} \right] \numberthis \label{equ_sol_p_t}
    \end{align*}   
    Note that $t', y_1,y_2$ are nonnegative so $p$ is well defined.

    With this we can conclude that for any $t<1$, there exist $p_t$ using equation~\eqref{equ_sol_p_t}, with it we can construct $\x'$ such that $GM(\x) = GM(\x')$ and $HM(\x') = t \cdot HM(\x)$. \\
    
    It is worth mentioning how the AM changes in this construction. Indeed,
    \[
    AM(\x') = \frac{1}{n}\sum_{i=1}^n x'_i = \frac{1}{n}(\sum_i x_i - x_1 - x_2 +px_1 + x_2/p) = AM(\x) + \frac{1}{n}((p_t-1)x_1+(1/p_t - 1)x_2).
    \]
    Interestingly, we get that the Arithmetic mean always increases. To see why consider equation~\eqref{equ_quad_p_t}
    \[
        (1/p -1)\frac{1}{x_1} + (p-1)\frac{1}{x_2} = S(1/t-1)
    \]
    Multiplying both sides by $x_1 x_2$ gives us
    \[
        (1/p-1)x_2 + (p-1)x_1 =S(1/t-1)x_1x_2 > 0.
    \]
    
\end{proof}

\begin{remark}
    It is worth noticing the behavior of, $\Delta AM = AM(x') - AM(x)$ . While $\frac{HM'}{HM} = t$,
    the difference in the AM is (see last equation in the last proof)
    \[
    \Delta AM = S x_1 x_2 \frac{1}{n}(1/t - 1) = \theta \left( \frac{1}{n}(1/t - 1) \right) \ \ \ \ \   (\text{we refer to $\boldsymbol{x}$ as constant.})
    \]
    For example, if $t = \theta(1/n)$ then,
    \[
    \frac{HM'}{HM} = \theta \left(\frac{1}{n} \right) \ , \  \Delta AM = \theta(1).
    \]
    So, although the value of the $HM$ is multiplied by a factor of $n$, the change in the value of $AM$ is relatively small. In other words, the decrease in the value of $HM$ outweighs the slight increase in the value of $AM$. 
\end{remark}

\subsection{The Lagrange Dual Problem of ERMP}\label{appendix_ERMP_dual_gap}

We first rewrite the ERMP problem as:

\begin{equation*}
\begin{array}{ll@{}ll}
\text{minimize}  & n \Tr X^{-1} - n &\\
\text{subject to}&  X = \sum_{l=1}^{m} g_lb_lb_l^T + (1/n)\boldsymbol{1}\boldsymbol{1}^T\\
                    & \boldsymbol{1}^Tg=1 , g \geq 0
\end{array}
\end{equation*}

with variables $g \in \R^m$ and $X = X^T \in \R^{n \times n}_+$. 
Define the dual variables $Z \in \boldsymbol{S}^{n}_+$ , $v \in \R$, for the equality constraints, and $\lambda \in \R^m$ for the non-negativity constraint $g \geq 0$. The Lagrangian is
\[
L(X,g,Z,v,\lambda) = n\Tr X^{-1} - n + \Tr [Z(X - \sum_{l=1}^{m} g_lb_lb_l^T + 1/n\boldsymbol{1}\boldsymbol{1}^T) ] + v(\boldsymbol{1}^T g-1) - \lambda^T g .
\]

Then, the dual function is,
\begin{align*}
h(Z,v,\lambda) &= \underset{ X \succeq 0, g}{\inf} L(X,g,Z,v,\lambda) \\
&= \underset{X \succeq 0}{\inf} \ \Tr [nX^{-1} + ZX] + \underset{g}{\inf} \ \left(\sum_{l=1}^{m} g_l(b_l^T Z b_l + v - \lambda_l)\right) -n -v -(1/n)\boldsymbol{1}^T Z \boldsymbol{1}
\end{align*}
\begin{equation*}
= \left\{ 
\begin{array}{ll@{}ll}
     & -v -(1/n)\boldsymbol{1}^T Z \boldsymbol{1} + 2\Tr(nZ)^{1/2} -n  \ \ \ \ \ &\text{if} \ \ b_l^T Z b_l +v  = \lambda_l \ , \ Z \succeq 0  \\
     &  -\infty & \text{otherwise}
\end{array}
\right.
\end{equation*}

The last line can be explained as follows. The left term $\Tr[nX^{-1} + ZX]$ is unbounded below, unless $Z \succeq 0$ (otherwise we could pick $X_1$ to be eigenvector corresponds to a negative eigenvalue of $Z$ multiplied by arbitrary large constant, and complete $X_{2...n}$ s.t $X$ is PSD). Now, if $Z \succ 0$ the unique $X$ that minimizes it is $X = (Z/n)^{-1/2}$ (this can be easily shown by taking the derivative w.r.t $X$) , hence its optimal value is,
\[
\Tr[nX^{-1} + ZX] = \Tr[n(Z/n)^{1/2} + Z(Z/n)^{-1/2}] = 2\Tr(nZ)^{1/2} .
\]
If $Z$ is PSD but not PD, this is still the optimal value as we can take a sequence of PD matrices $Z_i$ that converges to $Z$ and since this is a linear function of $Z$ over closed set, we get the same optimal value.

Additionally, the right term $\sum_{l=1}^{m} g_l(b_l^T Z b_l + v - \lambda_l)$ is also unbounded from below, unless $\forall l, \ b_l^T Z b_l + v - \lambda_l=0$. Otherwise, take $g=\alpha e_k$, such that, $b_k^T Z b_k + v - \lambda_k = c \neq 0$, leading to $\sum_{l=1}^{m} g_l(b_l^T Z b_l + v - \lambda_l) = \alpha \cdot c$ which is unbounded from below as function of $\alpha$.\\

The Lagrange dual problem is,
\begin{equation*}
    \begin{array}{ll@{}ll}
     & \text{maximize}  \ \ &h(Z,v,\lambda)  \\
     & \text{subject to} &  b_l^T Z b_l \leq v , Z \succeq 0, \lambda>0
\end{array}
\end{equation*}

Using the formula for $h(Z,v,\lambda)$ we derived above, and eliminating $\lambda$ (since it isn't taking part in the objective function), we obtain the following form:
\begin{equation*}
    \begin{array}{ll@{}ll}
     & \text{maximize} \ \ &  -v -(1/n)\boldsymbol{1}^T Z \boldsymbol{1} + 2\Tr(nZ)^{1/2} -n  \\
     & \text{subject to} &  b_l^T Z b_l \leq v , Z \succeq 0
\end{array}
\end{equation*}

This problem is another convex optimization problem with variables $Z , v$. 
It can be verified that, if $X^* = L^* + 1/n \boldsymbol{1}\boldsymbol{1}^T$ (where $L^* = \sum_{l=1}^{m} g_l^* b_l b_l^T$ the optimal weighted Laplacian), is the optimal solution for the primal ERMP, then 
\[
Z^* = n(X^*)^{-2} = n(L^* + 1/n \boldsymbol{1}\boldsymbol{1}^T)^{-2} \ , \ v^* = \underset{l}{\max} \ b_l^T Z^* b_l
\]
is the optimal solution for the dual problem (this is easily implied from our analysis). We can use it to further reduce the dual problem.
For any $g$, we have $(L_g + 1/n \boldsymbol{1}\boldsymbol{1}^T)^{-2} \boldsymbol{1} = \boldsymbol{1}$ (see equation~\eqref{equ_laplac_inv}). Thus, $Z^* \boldsymbol{1} = n\boldsymbol{1}$, meaning the optimal solution satisfies $Z \boldsymbol{1}=n\boldsymbol{1}$. So, we can add this constraint without changing the optimal value (which is $\cK^*$). The obtain dual problem is now,
\begin{equation*}
    \begin{array}{ll@{}ll}
     & \text{maximize}  \ \  & -v -(1/n)\boldsymbol{1}^T Z \boldsymbol{1} + 2\Tr(nZ)^{1/2} -n  \\
     & \text{subject to} &  b_l^T Z b_l \leq v , \\
     & & Z \succeq 0 \ ,\  Z \boldsymbol{1} = n \boldsymbol{1}
\end{array}
\end{equation*}
\\
Changing variables to $Y = Z - \boldsymbol{1}\boldsymbol{1}^T$, leads to the following form,
\begin{equation*}
    \begin{array}{ll@{}ll}
     & \text{maximize} \ \ & -v + 2\Tr(nY)^{1/2}  \\
     & \text{subject to} &  b_l^T Y b_l \leq v , \\
     & & Y \succeq 0 \ ,\  Y \boldsymbol{1} = 0
\end{array}
\end{equation*}
\\
Let $W = Y/v$ (note that from our constraints $v$ is a positive number), and our problem becomes:
\begin{equation*}
    \begin{array}{ll@{}ll}
     & \text{maximize} \ \  &-v + 2(nv)^{1/2}\Tr W^{1/2}  \\
     & \text{subject to} &  b_l^T W b_l \leq 1 , \\
     & & W \succeq 0 \ ,\  W \boldsymbol{1} = 0
\end{array}
\end{equation*}
\\
Now, the maximal value of $f(v) = -v + 2C \cdot v^{1/2}$ for positive constant $C$ is $v = C^2$. Taking $C = n^{1/2}\Tr W^{1/2}$, and the problem becomes
\begin{equation*}
    \begin{array}{ll@{}ll}
     & \text{maximize}  \ \ &n(\Tr W^{1/2})^2  \\
     & \text{subject to} &  b_l^T W b_l \leq 1 , \\
     & & W \succeq 0 \ ,\  W \boldsymbol{1} = 0
\end{array}
\end{equation*}
\\
Finally, changing variables to $V = W^{1/2}$, we can express the dual problem as,
\begin{equation*}
    \begin{array}{ll@{}ll}
     & \text{maximize}  \ \ &n(\Tr V)^2  \\
     & \text{subject to} &  b_l^T V^2 b_l = ||V b_l|| \leq 1 , \\  \numberthis \label{equ_ERMP_dual_prob} 
     & & V \succeq 0 \ ,\  V \boldsymbol{1} = 0
\end{array}
\end{equation*}

\subsection{SDP Formulation for The ERMP}\label{appendix_ermp_sdp}

Shown by \cite{saberi,boyd_convex_2004}, the ERMP can be formulated as an SDP:
\begin{equation*}
\begin{array}{ll@{}ll}
\text{minimize}  & n \Tr Y &\\
\text{subject to}&  g \geq 0 ,\ \boldsymbol{1}^T g=1 &\\
& \begin{bmatrix}
L_g + (1/n)\buni \buni^T & I \\
I & Y
\end{bmatrix}
\succeq 0
\end{array} 
\end{equation*}

Define, the coefficient matrix
\[
C \coloneqq n \begin{bmatrix}
0 & 0 \\
0 & -I
\end{bmatrix}
\]
and the constrains
\[
A \coloneqq \begin{bmatrix}
I & 0 
\end{bmatrix}, 
b \coloneqq \begin{bmatrix}
L_g + (1/n)\buni \buni^T & I 
\end{bmatrix} 
\]
Then the above formulation gets the form:
\begin{equation*}
\begin{array}{ll@{}ll}
\text{maximize}  & C \bullet X &\\
\text{subject to}&  g \geq 0 ,\ \boldsymbol{1}^T g=1 &\\
& A  X = b &\\
& X \succeq 0, X = X^T
\end{array} 
\end{equation*}
which is \cite{Haz08} formulation. 

\subsection{Optimality Criteria for SEV}\label{appendix_LSE_opt_crit}

In accordance to previous analysis, we will derive the optimality criteria for $f(g) = \Tr e^{L_g^+}$, using the familiar consition,
\[
g \ \text{is optimal for } f   \ \iff \nabla f(g)^T(e_l - g) \geq 0  \ \text{, for} \ l=1\dots m
\]

Let's start by computing the gradient of $f(g)$. Using the chain rule we get that:
\[
\frac{\partial f}{\partial g_l} = \Tr[ e^{L_g^+} \frac{\partial L_g^+}{\partial g_l} ]= -\Tr [e^{L_g^+} L_g^+ b_l b_l^T L_g^+]
\]
where we used equation~\eqref{equ_der_L_g_plus} for the gradient of $L_g^+$.\\

Next we want to compute \( \nabla f(g)^T g\). We will use the standard "trick":
\begin{align*}
&f(\alpha g) = \Tr[ e^{L_{\alpha g}^+} ] = \Tr[ e^{L_{g}^+ / \alpha} ] \\
&\implies \frac{\partial }{\partial \alpha}f(\alpha g) = \frac{\partial }{\partial \alpha} \Tr[ e^{L_{g}^+ / \alpha} ] = \Tr[ e^{L_{g}^+ / \alpha} \frac{\partial }{\partial \alpha}(L_{g}^+ / \alpha) ] = \Tr[ e^{L_{g}^+ / \alpha} L_g^+ (-1/\alpha^2) ]
\end{align*}
and setting $\alpha=1$, leads to
\[
\nabla f(g)^T g = \frac{\partial }{\partial \alpha}f(\alpha g) \Big|_{\alpha=1} = -\Tr [ e^{L_{g}^+} L_g^+] .
\]


plug those in the optimality condition, and we get the following form:
\[
-\Tr [e^{L_g^+} L_g^+ b_l b_l^T L_g^+] + \Tr[ e^{L_{g}^+} L_g^+] \geq 0
\]
and using the additive property of trace, we get equation~\eqref{equ_sev_opt_crit}:
\[
\Tr [ e^{L_g^+}L_g^+ (I - b_l b_l^T L_g^+)] \geq 0
\]

\subsection{Figures of Upper LT}

\begin{figure}[ht]
    \centering
    \begin{tikzpicture}
	\begin{pgfonlayer}{nodelayer}
		\node [style=none] (0) at (-8.75, 2.25) {};
		\node [style=none] (1) at (-5.25, 2.25) {};
		\node [style=none] (4) at (-2.25, 4) {};
		\node [style=none] (5) at (-1.725, 0) {};
		\node [style=none] (6) at (-8.65, 2.525) {$u$};
		\node [style=none] (7) at (-5.25, 2.525) {$v$};
		\node [style=none] (8) at (-6.75, 2.55) {$e_k$};
		\node [style={tree_node}] (9) at (-2, 4.25) {$T_1$};
		\node [style={tree_node}] (10) at (-1.5, 0.175) {$T_2$};
		\node [style=none] (11) at (-4.725, -1.2) {};
		\node [style={tree_node}] (12) at (-4.7, -1.225) {$T_d$};
		\node [style=none] (13) at (-8.75, 2.25) {};
		\node [style=none] (14) at (-11.75, 4) {};
		\node [style=none] (15) at (-9, -1.25) {};
		\node [style={tree_node}] (17) at (-11.5, 4.25) {$T'_1$};
		\node [style={tree_node}] (18) at (-12, 0.125) {$T'_2$};
		\node [style=none] (19) at (-11.825, 0.05) {};
		\node [style={tree_node}] (20) at (-8.85, -1.325) {$T'_d$};
		\node [style=none] (21) at (2.25, 1.75) {$\longrightarrow$};
		\node [style=none] (22) at (-3.05, 3.9) {\scriptsize $x_1$};
		\node [style=none] (23) at (-2.625, 0.175) {\scriptsize $x_2$};
		\node [style=none] (24) at (-4.5, -0.225) {\scriptsize $x_d$};
		\node [style={fill_node}] (25) at (-2.65, 3.75) {};
		\node [style={fill_node}] (26) at (-2.275, 0.375) {};
		\node [style={fill_node}] (27) at (-4.85, -0.425) {};
		\node [style=none] (28) at (-4.025, 3.375) {\small $l_1$};
		\node [style=none] (29) at (9.25, 2.25) {};
		\node [style=none] (30) at (12.75, 2.25) {};
		\node [style=none] (31) at (15.75, 4.25) {};
		\node [style=none] (32) at (18.025, 1.5) {};
		\node [style=none] (33) at (9.35, 2.525) {$u$};
		\node [style=none] (34) at (12.75, 2.525) {$v$};
		\node [style=none] (35) at (11.25, 2.55) {$e_k$};
		\node [style={tree_node}] (36) at (16, 4.5) {$T_1$};
		\node [style={tree_node}] (37) at (18.25, 1.675) {$T_2$};
		\node [style=none] (38) at (13.275, -1.2) {};
		\node [style={tree_node}] (39) at (13.3, -1.225) {$T_d$};
		\node [style=none] (40) at (9.25, 2.25) {};
		\node [style=none] (41) at (6.25, 4) {};
		\node [style=none] (42) at (9, -1.25) {};
		\node [style={tree_node}] (43) at (6.5, 4.25) {$T_1$};
		\node [style={tree_node}] (44) at (6, 0.125) {$T_2$};
		\node [style=none] (45) at (6.175, 0.05) {};
		\node [style={tree_node}] (46) at (9.15, -1.325) {$T_d$};
		\node [style=none] (47) at (14.975, 4.075) {\scriptsize $x_1$};
		\node [style=none] (48) at (17.075, 1.825) {\scriptsize $x_2$};
		\node [style=none] (49) at (13.5, -0.2) {\scriptsize $x_d$};
		\node [style={fill_node}] (50) at (15.375, 3.95) {};
		\node [style={fill_node}] (51) at (17.475, 2) {};
		\node [style={fill_node}] (52) at (13.15, -0.425) {};
		\node [style=none] (53) at (13.975, 3.425) {\small $l_1$};
		\node [style=none] (54) at (-5.75, -2.5) {};
		\node [style=none] (55) at (-0.75, -2.5) {};
		\node [style=none] (56) at (-3.25, -3.5) {$n_v$};
		\node [style=none] (57) at (16.625, -0.2) {};
		\node [style={tree_node}] (58) at (16.65, -0.225) {$T_3$};
		\node [style=none] (59) at (15.475, 0.05) {\scriptsize $x_3$};
		\node [style={fill_node}] (60) at (15.95, 0.2) {};
	\end{pgfonlayer}
	\begin{pgfonlayer}{edgelayer}
		\draw [style={selected_edge}] (0.center) to (1.center);
		\draw (4.center) to (1.center);
		\draw (1.center) to (5.center);
		\draw (1.center) to (11.center);
		\draw [style={dash_lines}, in=0, out=-120, looseness=1.25] (5.center) to (11.center);
		\draw (14.center) to (13.center);
		\draw (13.center) to (15.center);
		\draw (13.center) to (19.center);
		\draw [style={dash_lines}, in=-90, out=-150] (15.center) to (19.center);
		\draw [style={selected_edge}] (29.center) to (30.center);
		\draw (31.center) to (30.center);
		\draw (30.center) to (38.center);
		\draw (41.center) to (40.center);
		\draw (40.center) to (42.center);
		\draw (40.center) to (45.center);
		\draw [style={dash_lines}, in=-90, out=-150] (42.center) to (45.center);
		\draw (50) to (32.center);
		\draw [style={vertical_curly_braket}] (54.center) to (55.center);
		\draw (30.center) to (58);
		\draw [style={dash_lines}, bend left] (58) to (39);
	\end{pgfonlayer}
\end{tikzpicture}
    \caption{Description of upper LT for first case }
    \label{fig_upper_ET_case1}
\end{figure}
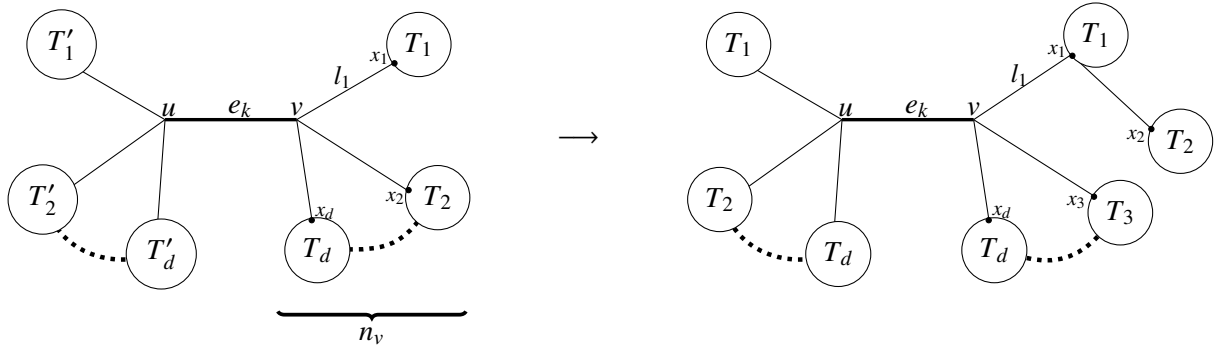

\begin{figure}[ht]
    \centering
    \begin{tikzpicture}
	\begin{pgfonlayer}{nodelayer}
		\node [style=none] (0) at (-8.75, 2.25) {};
		\node [style=none] (1) at (-5.25, 2.25) {};
		\node [style=none] (4) at (-2.25, 4) {};
		\node [style=none] (5) at (-1.725, 0) {};
		\node [style=none] (6) at (-8.65, 2.55) {$u$};
		\node [style=none] (7) at (-5.25, 2.55) {$v$};
		\node [style=none] (8) at (-6.75, 2.7) {$e_k$};
		\node [style={tree_node}] (9) at (-2, 4.25) {$T_1$};
		\node [style={tree_node}] (10) at (-1.5, 0.175) {$T_2$};
		\node [style=none] (11) at (-4.725, -1.2) {};
		\node [style={tree_node}] (12) at (-4.7, -1.225) {$T_d$};
		\node [style=none] (13) at (-8.75, 2.25) {};
		\node [style=none] (14) at (-11.75, 2.25) {};
		\node [style=none] (21) at (2.25, 1.75) {$\longrightarrow$};
		\node [style=none] (22) at (-3.05, 3.9) {\scriptsize $x_1$};
		\node [style=none] (23) at (-2.65, 0.175) {\scriptsize $x_2$};
		\node [style=none] (24) at (-4.5, -0.15) {\scriptsize $x_d$};
		\node [style={fill_node}] (25) at (-2.65, 3.75) {};
		\node [style={fill_node}] (26) at (-2.3, 0.35) {};
		\node [style={fill_node}] (27) at (-4.85, -0.425) {};
		\node [style=none] (28) at (-4.025, 3.375) {\small $l_1$};
		\node [style=none] (29) at (9.25, 2.25) {};
		\node [style=none] (30) at (12.75, 2.25) {};
		\node [style=none] (31) at (16.5, 4.5) {};
		\node [style=none] (32) at (16.275, 0) {};
		\node [style=none] (33) at (9.35, 2.55) {$u$};
		\node [style=none] (34) at (12.75, 2.55) {$v$};
		\node [style=none] (35) at (11.25, 2.7) {$e_k$};
		\node [style={tree_node}] (36) at (16.75, 4.75) {$T_1$};
		\node [style={tree_node}] (37) at (16.5, 0.175) {$T_3$};
		\node [style=none] (38) at (13.275, -1.2) {};
		\node [style={tree_node}] (39) at (13.3, -1.225) {$T_d$};
		\node [style=none] (40) at (9.25, 2.25) {};
		\node [style=none] (47) at (15.7, 4.375) {\scriptsize $x_1$};
		\node [style=none] (48) at (15.325, 0.275) {\scriptsize $x_3$};
		\node [style=none] (49) at (13.5, -0.15) {\scriptsize $x_d$};
		\node [style={fill_node}] (50) at (16.125, 4.25) {};
		\node [style={fill_node}] (51) at (15.7, 0.375) {};
		\node [style={fill_node}] (52) at (13.15, -0.425) {};
		\node [style=none] (53) at (14.225, 3.625) {\small $l_1$};
		\node [style=none] (54) at (-11.575, 2.575) {$y$};
		\node [style=none] (55) at (-10.25, 2.725) {$l'$};
		\node [style=none] (59) at (6.1, 2.25) {};
		\node [style=none] (60) at (6.2, 2.525) {$y$};
		\node [style=none] (61) at (7.6, 2.725) {$l'$};
		\node [style={fill_node}] (62) at (9.25, 2.25) {};
		\node [style={fill_node}] (63) at (-8.75, 2.25) {};
		\node [style={fill_node}] (64) at (6.05, 2.25) {};
		\node [style={fill_node}] (65) at (-11.75, 2.25) {};
		\node [style={tree_node}] (66) at (-12.625, 2.2) {$T_0$};
		\node [style={tree_node}] (67) at (5.175, 2.175) {$T_0$};
		\node [style=none] (68) at (-6, -2.5) {};
		\node [style=none] (69) at (-1, -2.5) {};
		\node [style=none] (71) at (-3.5, -3.75) {$T_v$};
		\node [style=none] (72) at (-8.25, 1) {};
		\node [style=none] (73) at (-13.5, 1) {};
		\node [style=none] (74) at (-11, 0) {$n'$};
		\node [style=none] (75) at (-15.25, 1.75) {};
		\node [style=none] (76) at (18.475, 1.25) {};
		\node [style={tree_node}] (77) at (18.7, 1.425) {$T_2$};
		\node [style={fill_node}] (78) at (18, 1.85) {};
		\node [style=none] (79) at (17.575, 1.65) {\scriptsize $x_2$};
	\end{pgfonlayer}
	\begin{pgfonlayer}{edgelayer}
		\draw [style={selected_edge}] (0.center) to (1.center);
		\draw (4.center) to (1.center);
		\draw (1.center) to (5.center);
		\draw (1.center) to (11.center);
		\draw [style={dash_lines}, in=0, out=-120, looseness=1.25] (5.center) to (11.center);
		\draw (14.center) to (13.center);
		\draw [style={selected_edge}] (29.center) to (30.center);
		\draw (31.center) to (30.center);
		\draw (30.center) to (38.center);
		\draw [style={dash_lines}, in=0, out=-120, looseness=1.25] (32.center) to (38.center);
		\draw (59.center) to (40.center);
		\draw [style={vertical_curly_braket}] (68.center) to (69.center);
		\draw [style={vertical_curly_braket}] (73.center) to (72.center);
		\draw (30.center) to (32.center);
		\draw (76.center) to (50);
	\end{pgfonlayer}
\end{tikzpicture}
    \caption{Description of upper LT for case 2.1 }
    \label{fig_upper_ET_case2.1}
\end{figure}
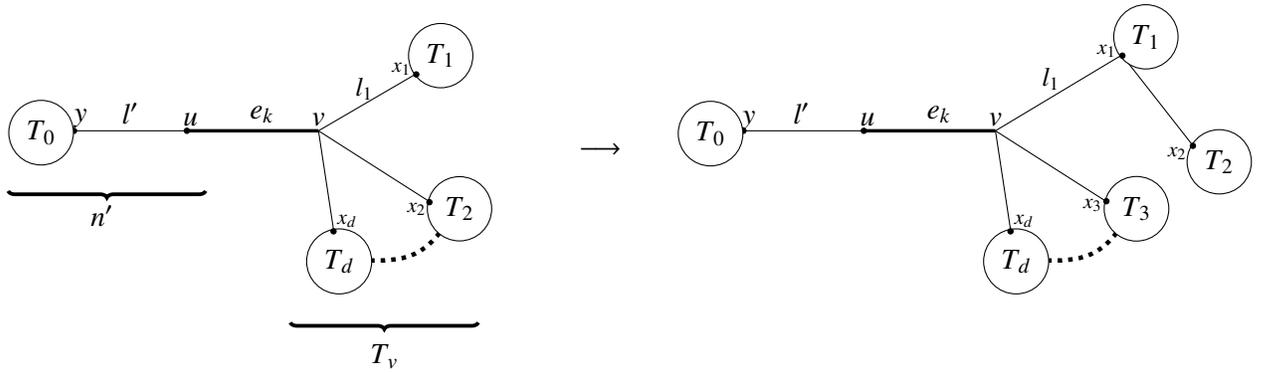


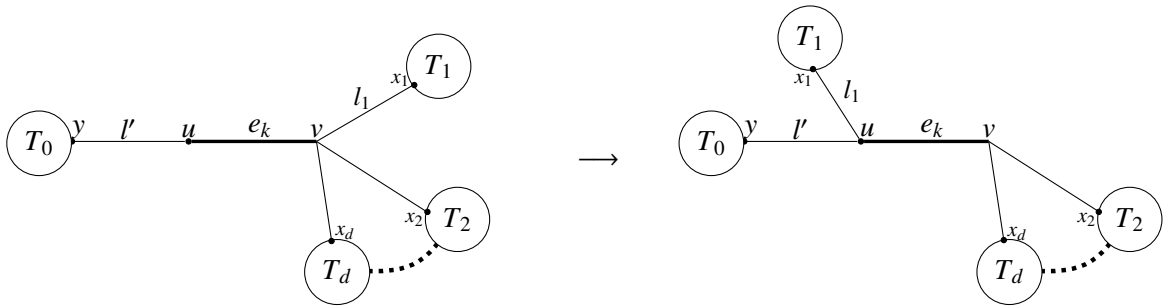
\begin{figure}[ht]
    \centering
    \begin{tikzpicture}
	\begin{pgfonlayer}{nodelayer}
		\node [style=none] (0) at (-8.75, 2.25) {};
		\node [style=none] (1) at (-5.25, 2.25) {};
		\node [style=none] (4) at (-2.25, 4) {};
		\node [style=none] (5) at (-1.725, 0) {};
		\node [style=none] (6) at (-8.65, 2.525) {$u$};
		\node [style=none] (7) at (-5.25, 2.525) {$v$};
		\node [style=none] (8) at (-6.75, 2.6) {$e_k$};
		\node [style={tree_node}] (9) at (-2, 4.25) {$T_1$};
		\node [style={tree_node}] (10) at (-1.5, 0.175) {$T_2$};
		\node [style=none] (11) at (-4.725, -1.2) {};
		\node [style={tree_node}] (12) at (-4.7, -1.225) {$T_d$};
		\node [style=none] (13) at (-8.75, 2.25) {};
		\node [style=none] (14) at (-11.75, 2.25) {};
		\node [style=none] (21) at (2.25, 1.75) {$\longrightarrow$};
		\node [style=none] (22) at (-3.025, 3.9) {\scriptsize $x_1$};
		\node [style=none] (23) at (-2.65, 0.175) {\scriptsize $x_2$};
		\node [style=none] (24) at (-4.5, -0.175) {\scriptsize $x_d$};
		\node [style={fill_node}] (25) at (-2.65, 3.75) {};
		\node [style={fill_node}] (26) at (-2.3, 0.375) {};
		\node [style={fill_node}] (27) at (-4.85, -0.4) {};
		\node [style=none] (28) at (-4.025, 3.4) {\small $l_1$};
		\node [style=none] (54) at (-11.575, 2.575) {$y$};
		\node [style=none] (55) at (-10.25, 2.575) {$l'$};
		\node [style={fill_node}] (63) at (-8.65, 2.25) {};
		\node [style={fill_node}] (65) at (-11.75, 2.25) {};
		\node [style={tree_node}] (66) at (-12.625, 2.2) {$T_0$};
		\node [style=none] (67) at (9.125, 2.25) {};
		\node [style=none] (68) at (12.625, 2.25) {};
		\node [style=none] (69) at (7.625, 4.75) {};
		\node [style=none] (70) at (16.15, 0) {};
		\node [style=none] (71) at (9.4, 2.55) {$u$};
		\node [style=none] (72) at (12.625, 2.525) {$v$};
		\node [style=none] (73) at (11.125, 2.6) {$e_k$};
		\node [style={tree_node}] (74) at (7.875, 5) {$T_1$};
		\node [style={tree_node}] (75) at (16.375, 0.175) {$T_2$};
		\node [style=none] (76) at (13.15, -1.2) {};
		\node [style={tree_node}] (77) at (13.175, -1.225) {$T_d$};
		\node [style=none] (78) at (9.125, 2.25) {};
		\node [style=none] (79) at (6.125, 2.25) {};
		\node [style=none] (80) at (7.7, 3.9) {\scriptsize $x_1$};
		\node [style=none] (81) at (15.225, 0.175) {\scriptsize $x_2$};
		\node [style=none] (82) at (13.375, -0.2) {\scriptsize $x_d$};
		\node [style={fill_node}] (83) at (7.95, 4.175) {};
		\node [style={fill_node}] (84) at (15.55, 0.375) {};
		\node [style={fill_node}] (85) at (13.025, -0.375) {};
		\node [style=none] (86) at (9, 3.6) {\small $l_1$};
		\node [style=none] (87) at (6.3, 2.575) {$y$};
		\node [style=none] (88) at (7.625, 2.575) {$l'$};
		\node [style={fill_node}] (89) at (9.225, 2.25) {};
		\node [style={fill_node}] (90) at (6.125, 2.25) {};
		\node [style={tree_node}] (91) at (5.25, 2.2) {$T_0$};
	\end{pgfonlayer}
	\begin{pgfonlayer}{edgelayer}
		\draw (4.center) to (1.center);
		\draw (1.center) to (5.center);
		\draw (1.center) to (11.center);
		\draw [style={dash_lines}, in=0, out=-120, looseness=1.25] (5.center) to (11.center);
		\draw (14.center) to (13.center);
		\draw (68.center) to (70.center);
		\draw (68.center) to (76.center);
		\draw [style={dash_lines}, in=0, out=-120, looseness=1.25] (70.center) to (76.center);
		\draw (79.center) to (78.center);
		\draw (89) to (69.center);
		\draw [style={selected_edge}] (63) to (1.center);
		\draw [style={selected_edge}] (89) to (68.center);
	\end{pgfonlayer}
\end{tikzpicture}
    \caption{Description of upper LT for case 2.2 }
    \label{fig_upper_ET_case2.2}
\end{figure}

\end{document}